\documentclass[conference]{IEEEtran}
\IEEEoverridecommandlockouts
\usepackage{graphicx}
\usepackage{amsfonts}
\usepackage{textcomp}
\usepackage{xcolor}
\usepackage{fancyhdr}
\usepackage{soul}
\usepackage{multirow}
\usepackage{diagbox} 
\usepackage{booktabs} 
  
\usepackage[noend]{algpseudocode}  
\usepackage{algorithm}
\usepackage{url}
\usepackage{makecell}
\usepackage{threeparttable}
\usepackage{amsmath}
\usepackage{amsthm}
\usepackage{subcaption}
\usepackage{color}
\usepackage{tikz}
\usepackage{comment}
\usepackage{balance}
\usepackage{pifont}
\usepackage{siunitx}
\newcommand{\cmark}{\ding{51}}%
\newcommand{\xmark}{\ding{55}}%

\newtheorem{theorem}{Theorem}

\newtheorem{lemma}{Lemma}

\usepackage{enumitem}
\setenumerate[1]{itemsep=0pt,partopsep=0pt,parsep=\parskip,topsep=5pt}
\setitemize[1]{itemsep=0pt,partopsep=0pt,parsep=\parskip,topsep=5pt}
\setdescription{itemsep=0pt,partopsep=0pt,parsep=\parskip,topsep=5pt}
\newcommand{\fangyu}[1]{{\color{black}#1}}

\newcommand{\edits}[1]{{\color{black}#1}}

\newcommand{\arxiv}[1]{\textcolor{black}{#1}}

\usepackage[font=small,labelfont=bf]{caption}

\usepackage{amsmath}
\def\BibTeX{{\rm B\kern-.05em{\sc i\kern-.025em b}\kern-.08em
    T\kern-.1667em\lower.7ex\hbox{E}\kern-.125emX}}
\begin{document}

\title{Scaling Blockchain Consensus via\\a Robust Shared Mempool
\thanks{\edits{$\dagger$These authors have contributed equally to this work. Corresponding author: Jianyu Niu.
}}}

\author{\IEEEauthorblockN{Fangyu Gai$^{1,\dagger}$, Jianyu Niu$^{2,\dagger}$, 
Ivan Beschastnikh$^3$, Chen Feng$^1$, Sheng Wang$^4$}
\IEEEauthorblockA{$^1$\{fangyu.gai, chen.feng\}@ubc.ca \hspace{0.3cm}$^2$niujy@sustech.edu.cn \hspace{0.3cm} $^3$bestchai@cs.ubc.ca\hspace{0.3cm}$^4$sh.wang@alibaba-inc.com}
\IEEEauthorblockA{University of British Columbia ($^1$Okanagan Campus, $^3$Vancouver Campus) \\ $^2$Southern University of Science and Technology \hspace{0.3cm} $^4$Alibaba Group}
}

\maketitle
\thispagestyle{plain}
\pagestyle{plain}

\begin{abstract}
Leader-based Byzantine fault-tolerant (BFT) consensus protocols used by permissioned blockchains have limited scalability and robustness.
To alleviate the leader bottleneck in BFT consensus, we introduce \emph{Stratus}, a robust shared mempool protocol that decouples transaction distribution from consensus.
Our idea is to have replicas disseminate transactions in a distributed manner and have the leader only propose transaction ids.
Stratus uses a provably available broadcast (PAB) protocol to ensure the availability of the referenced transactions. 
To deal with unbalanced load across replicas, Stratus adopts a distributed load balancing protocol.

We implemented and evaluated Stratus by integrating it with state-of-the-art BFT-based blockchain protocols.
Our evaluation of these protocols in both LAN and WAN settings shows that Stratus-based protocols achieve $5\times$ to $20\times$ higher throughput than their native counterparts in a network with hundreds of replicas.
In addition, the performance of Stratus degrades gracefully in the presence of network asynchrony, Byzantine attackers, and unbalanced workloads.

\end{abstract}


\section{Introduction} \label{sec:intro}
The emergence of blockchain technology has revived interest in Byzantine fault-tolerant (BFT) systems~\cite{falcondb,blockchaindb,hyperledger,slimchain,fireledger}.
Unlike traditional distributed databases, BFT systems (or blockchains) provide data provenance and allow federated data processing in untrusted and hostile environments~\cite{dbandblockchain,basil}.
This enables a rich set of decentralized applications, in e.g., finance~\cite{finance}, gaming~\cite{kitties}, healthcare~\cite{healthcare}, and social media~\cite{steemit}.
Many companies and researchers are seeking to build enterprise-grade blockchain systems~\cite{Tendermint,diem,Dapper,algorand} to provide Internet-scale decentralized services~\cite{decentralizedservice}.

The core of a blockchain system is the BFT consensus protocol, which allows distrusting parties to replicate and order a sequence of transactions.
Many BFT consensus protocols~\cite{sbft,tendermint-gossip,hotstuff,streamlet} adopted by permissioned blockchains follow the classic leader-based design of PBFT~\cite{pbft}: only the leader node determines the order to avoid conflicts.
We call such protocols leader-based BFT protocols, or \underline{L}BFT.

In the normal case (Byzantine-free), an LBFT consensus instance roughly consists of a proposing phase and a commit phase.
In the proposing phase, the leader pulls transactions from its local transaction pool (or mempool), forms a proposal, and broadcasts the proposal to the other replicas.
On receiving a proposal, replicas verify the proposal content before entering the commit phase.
In the commit phase, the leader coordinates multiple rounds of message exchanges to ensure that all correct replicas commit the same proposal at the same position.
If the leader behaves in a detectable Byzantine manner, a view-change sub-protocol will be triggered to replace the leader with one of the replicas.

A key scalability challenge for LBFT is the leader bottleneck.
Since the proposing and commit phases are both handled by the leader, adding replicas increases the load on the leader and reduces performance.
For example, in a LAN environment, the throughput of LBFT protocols drops from 120K tps (transaction per second) with 4 replicas to 20K tps with 64 replicas, while the transaction latency surges from $9$ milliseconds to $3$ seconds~\cite{bamboo}.
This has also been documented by other work~\cite{narwal,rcc,leopard}.

Prior work has focused on increasing LBFT performance by improving the \emph{commit phase}, e.g., reducing message complexity~\cite{hotstuff}, truncating communication rounds~\cite{zyzzyva}, and enhancing tolerance to Byzantine faults~\cite{fastbft,sft}.
Recent works~\cite{narwal,leopard} reveal that a more significant factor limiting LBFT's scalability lies in the \emph{proposing phase}, in which a proposal with batched transaction data (e.g., 10 MB) is disseminated by the single leader node, whereas messages exchanged in the commit phase (e.g., signatures, hashes) are much smaller (e.g., 100 Byte).
\arxiv{Formal analysis in Appendix~\ref{sec:leader-bottleneck-analysis} shows that reducing the message complexity of the commit phase cannot address this scalability issue.}

More broadly, previous works to address the leader bottleneck have proposed \textit{horizontal scaling} or sharding the blockchain into shards that concurrently run consensus~\cite{omniledger,sharper,monoxide,blockchaindb}.
These approaches require a large network to ensure safety~\cite{gearbox} and demand meticulous coordination for cross-shard transactions. By contrast, \textit{vertical scaling} approaches employ hierarchical schemes to send out messages and collect votes~\cite{pigpaxos,kauri}. Unfortunately, this increases latency and requires complex re-configuration to deal with faults.

In this paper, we follow neither of the above strategies.
Instead, we introduce the \textit{shared mempool} (SMP) abstraction, which decouples transaction distribution from consensus, leaving consensus with the job of ordering transaction ids.
SMP allows every replica to accept and disseminate client transactions so that the leader only needs to order transaction ids.
Applying SMP reaps the following benefits.
\emph{First}, SMP reduces the proposal size and increases throughput.
\emph{Second}, SMP decouples the transaction synchronization from ordering so that non-leader replicas can help with transaction distribution.
\emph{Lastly}, SMP can be integrated into existing systems without changing the consensus core.

\fangyu{SMP has been used to improve scalability~\cite{s-paxos,narwal,leopard}, but prior work has passed over two challenges.
Challenge 1: ensuring the availability of transactions referenced in a proposal.
When a replica receives a proposal, its local mempool may not contain all the referenced transactions.
These missing transactions prevent consensus from entering the commit phase, which may cause frequent view-changes (Section~\ref{sec:eval-byz}).
Challenge 2: dealing with unbalanced load across replicas.
SMP distributes the load from the leader and lets each replica disseminate transactions.
But, real workloads are highly skewed~\cite{skewed}, overwhelming some replicas and leaving others under-utilized (Section~\ref{sec:eval-unbalance}).
Existing SMP protocols ignore this and assume that each client sends transactions to a uniformly random replica~\cite{leopard,narwal,s-paxos}, but this assumption does not hold in \edits{practical} deployments~\cite{ethna,info-propagation,txprobe,Miller2015DiscoveringB}.}

\fangyu{We address these challenges with \emph{Stratus}, an SMP implementation that scales leader-based blockchains to hundreds of nodes.
Stratus introduces a \textit{provably available broadcast} (PAB) primitive to ensure the availability of transactions referenced in a proposal.
With PAB, consensus can safely enter the commit phase and not block on missing transactions.
To deal with unbalanced workloads, Stratus uses a distributed load-balancing (DLB) co-designed with PAB.
DLB dynamically estimates a replica's workload and capacity so that overloaded replicas can forward their excess load to under-utilized replicas.}
To summarize, we make the following contributions:
\begin{itemize}[leftmargin=*]
    \item We introduce and study a shared mempool abstraction that decouples network-based synchronization from ordering for leader-based BFT protocols.
    To the best of our knowledge, we are the first to study this abstraction explicitly.
    \item To ensure the availability of transactions, we introduce a broadcast primitive called PAB, which allows replicas to process proposals without waiting for transaction data.
    \item To balance load across replicas, we introduce a distributed load-balancing protocol co-designed with PAB, which allows busy replicas to transfer their excess load to under-utilized replicas.
    \item We implemented Stratus and integrated it with HotStuff~\cite{hotstuff},  Streamlet~\cite{streamlet}, and PBFT~\cite{pbft}.
    We show that Stratus-based protocols substantially outperform the native protocols in throughput, reaching up to $5\times$ and $20\times$ in typical LANs and WANs with 128 replicas.
    Under unbalanced workloads, Stratus achieves up to $10\times$ more throughput.
\end{itemize}


\section{Related Work}\label{sec:related-work}

One classic approach that relieves the load on the leader is \textit{horizontal scaling}, or sharding~\cite{sharper,blockchaindb,omniledger}.
However, using sharding in BFT consensus requires inter-shard and intra-shard consensus, which adds extra complexity to the system.
An alternative, \textit{vertical scaling} technique has been used in PigPaxos~\cite{pigpaxos}, which replaced direct communication between a Paxos leader and replicas with relay-based message flow.

Recently, many scalable designs have been proposed to bypass the leader bottleneck.
Algorand~\cite{algorand} can scale up to tens of thousands of replicas using Verifiable Random Functions (VRFs)~\cite{vrf} and a novel Byzantine agreement protocol called BA$\star$.
For each consensus instance, a committee is randomly selected via VRFs to reach consensus on the next set of transactions.
Some protocols such as HoneyBadger~\cite{honeybadger} and Dumbo~\cite{dumbo} adopt a leader-less design in which all the replicas contribute to a proposal.
They are targeting on consensus problems under asynchronous networks, while our proposal is for partially synchronous networks.
Multi-leader BFT protocols~\cite{rcc,fnf,mir-bft} have multiple consensus instances run concurrently, each led by a different leader.
\edits{Multi-leader BFT protocols such as MirBFT~\cite{mir-bft} and RCC~\cite{rcc} use multiple consensus instances that are run concurrently by different leaders.
These protocols follow a monolithic approach and introduce mechanisms in the view-change procedure to deal with the ordering across different instances and during failures.
These additions render a BFT system more error-prone and inefficient in recovery.
Stratus-enabled protocols are agnostic to the view-change since Stratus does not modify the consensus core.}

Several proposals address the leader bottleneck in BFT, and we compare these in Table~\ref{table:comparison}.
Tendermint uses gossip to shed the load from the leader.
Specifically, a block proposal is divided into several parts and each part is gossiped into the network.
Replicas reconstruct the whole block after receiving all parts of the block.
The most recent work, Kauri~\cite{kauri}, follows the \textit{vertically scaling} approach by arranging nodes in a tree to propagate transactions and collect votes.
It leverages a pipelining technique and a novel re-configuration strategy to overcome the disadvantages of using a tree structure. 
However, Kauri's fast re-configuration requires a large fan-out parameter (that is at least larger than the number of expected faulty replicas), which constrains its ability to load balance.
In general, tree-based approaches increase latency and require complex re-configuration strategies to deal with faults.

\fangyu{To our knowledge, S-Paxos~\cite{s-paxos} is the first consensus protocol to use a shared Mempool (SMP) to resolve the leader bottleneck. S-Paxos is not designed for Byzantine failures.}
Leopard~\cite{leopard} and Narwhal~\cite{narwal} utilize SMP to separate transaction dissemination from consensus and are most similar to our work.
Leopard modifies the consensus core of PBFT to allow different consensus instances to execute in parallel, since transactions may not be received in the order that proposals are proposed.
However, Leopard does not guarantee that the referenced transactions in a proposal will be available. It also does not scale well when the load across replicas is unbalanced.
Narwhal~\cite{narwal} is a DAG-based Mempool protocol.
It employs reliable broadcast (RB)~\cite{rb} to reliably disseminate transactions and uses a DAG to establish a causal relationship among blocks.
Narwhal can make progress even if the consensus protocol is stuck.
However, RB incurs quadratic message complexity and Narwhal only scales well when the nodes running the Mempool and nodes running the consensus are located on separate machines.
\fangyu{Our work differs from prior systems by contributing (1) an efficient and resilient broadcast primitive, along with (2) a co-designed load balancing mechanism to handle uneven workloads.}

\begin{table}[t]
  \caption{Existing work addressing the leader bottleneck.}
  \label{table:comparison}
\begin{center}
\footnotesize
\begin{tabular}{ | m{1.7cm} | m{1.41cm}| m{1.3cm} | m{.95cm} | m{1.35cm} |} 
  \hline
  \textbf{Protocol}& \textbf{Approach} & \textbf{Avail. guarantee} & \textbf{Load balance} & \textbf{Message complexity} \\ 
  \hline
  Tendermint~\cite{tendermint-gossip} &\hfil Gossip &\hfil \cmark &\hfil \cmark &\hfil $O(n^2)$ \\ 
  Kauri~\cite{kauri} &\hfil Tree &\hfil \cmark &\hfil \checkmark\kern-1.1ex\raisebox{.7ex}{\rotatebox[origin=c]{125}{--}} &\hfil $O(n)$ \\ 
  Leopard~\cite{leopard} &\hfil SMP &\hfil \xmark &\hfil \xmark &\hfil $O(n)$ \\ 
  Narwhal~\cite{narwal} &\hfil SMP &\hfil \cmark &\hfil \xmark &\hfil $O(n^2)$ \\ 
  \fangyu{MirBFT~\cite{mir-bft}} &\hfil \fangyu{Multi-leader} &\hfil \fangyu{\cmark} &\hfil \fangyu{\xmark} &\hfil \fangyu{$O(n^2)$} \\ 
  \hline
  \textbf{Stratus} &\hfil SMP &\hfil \cmark &\hfil \cmark &\hfil $O(n)$ \\ 
  \hline
\end{tabular}
\end{center}
\end{table}

\section{Shared Mempool Overview}\label{sec:shared-mempool}
We propose a \textit{shared mempool} (SMP) abstraction that decouples transaction dissemination from consensus to replace the original mempool in leader-based BFT protocols.
This decoupling idea enables us to use off-the-shelf consensus protocols rather than designing a scalable protocol from scratch.

\subsection{System Model}\label{sec:system-model}
We consider two roles in the BFT protocol: leader and replica.
A replica can become a \emph{leader} replica via view-changes or leader-rotation.
We inherit the Byzantine threat model and communication model from general BFT protocols~\cite{pbft,hotstuff}.
In particular, there are $N\geq3f+1$ replicas in the network and at most $f$ replicas are Byzantine.
The network is partially synchronous, whereby a known bound $\Delta$ on message transmission holds after some unknown Global Stabilization Time (GST)~\cite{PartialSynchrony}.

We consider external clients that issue transactions to the system.
We assume that each transaction has a unique ID and that every client knows about all the replicas (e.g., their IP addresses).
\fangyu{We also assume that each replica knows, or can learn, the leader for the current view.
Clients can select replicas based on network delay measurements, a random hash function, or another preference.
Byzantine replicas can censor transactions, however, so a client may need to switch to another replica (using a timeout mechanism) until a correct replica is found.
We assume that messages sent in our system are cryptographically signed and authenticated.
The adversary cannot break these signatures.

We futher assume that \textit{clients send each transaction to exactly one replica, but they are free to choose the replica for each transaction.}
Byzantine clients can perform a \textit{duplicate attack} by sending identical transactions to multiple replicas.
We consider these attacks out of scope.
In future work we plan to defend against these attacks using the bucket and transaction partitioning mechanism from MirBFT~\cite{mir-bft}.}


\subsection{Abstraction}\label{sec:abstraction}

A mempool protocol is a built-in component in a consensus protocol, running at every replica.
The mempool uses the \textit{ReceiveTx(tx)} primitive to receive transactions from clients and store them in memory (or to disk, if necessary).
If a replica becomes the leader, it calls the \textit{MakeProposal()} primitive to pull transactions from the mempool and constructs a proposal for the subsequent consensus process.
In most existing cryptocurrencies and permissioned blockchains~\cite{bitcoin,diem,tendermint-gossip}, the \textit{MakeProposal()} primitive generates a full proposal that includes all the transaction data.
As such, the leader bears the responsibility for transaction distribution and consensus coordination, leading to the leader bottleneck.
\arxiv{See our analysis in Appendix~\ref{sec:leader-bottleneck-analysis}.}

To relieve the leader's burden of distributing transaction data, we propose a \textit{shared mempool} (SMP) abstraction, which has been used in the previous works~\cite{graphene,leopard,narwal}, but has not been systematically studied.
The SMP abstraction enables the transaction data to be first disseminated among replicas, and then small-sized proposals containing only transaction ids are produced by the leader for replication. 
In addition, transaction data can be broadcast in batches with a unique id for each batch.
This further reduces the proposal size.
\arxiv{See our analysis in Appendix~\ref{sec:shared-mempool-analysis}.}
The SMP abstraction requires the following properties:

\textbf{SMP-Inclusion:} \textit{If a transaction is received and verified by a correct replica, then it is eventually included in a proposal}.

\textbf{SMP-Stability:} \textit{If a transaction is included in a proposal by a correct leader, then every correct replica eventually receives the transaction}.

The above two liveness properties ensure that a valid transaction is eventually replicated among correct replicas.
Particularly, \textbf{SMP-Inclusion} ensures that every valid transaction is eventually proposed while \textbf{SMP-Stability}, first mentioned in~\cite{s-paxos}, ensures that every proposed transaction is eventually available at all the correct replicas.
The second property makes SMP non-trivial to implement in a Byzantine environment; we elaborate on this in Section~\ref{sec:challenges}.
We should note that a BFT consensus protocol needs to ensure that all the correct replicas maintain the same history of transaction, or safety.
Using SMP does not change the order of committed transactions.
Thus, the safety of the consensus protocol is always maintained.

\subsection{Primitives and Workflow}

The implementation of the SMP abstraction modifies the two primitives \textit{ReceiveTx(tx)} and \textit{MakeProposal()} used in the traditional Mempool and adds two new primitives \textit{ShareTx(tx)} and \textit{FillProposal(p)} as follows:
\begin{itemize}[leftmargin=*]
    \item \textit{ReceiveTx(tx)} is used to receive an incoming $tx$ from a client or replica, and stores it in memory (or disk if necessary).
    \item \textit{ShareTx(tx)} is used to distribute $tx$ to other replicas.
    \item \textit{MakeProposal()} is used by the leader to pull transactions from the local mempool and construct a proposal with their ids.
    \item \textit{FillProposal(p)} is used when receiving a new proposal $p$.
    It pulls transactions from the local mempool according to the transaction ids in $p$ and fills it into a full proposal.
    It returns missing transactions if there are any.
\end{itemize}

Next, we show how these primitives work in an order-execute (OE) model, where transactions are first ordered through a consensus engine (using leader-based BFT consensus protocols) and then sent to an executor for execution.
We argue that while for simplicity our description hinges on an OE model, the principles could also be used in an execute-order-validate (EOV) model that is adopted by Hyperledger~\cite{hyperledger}.

We use two primitives from the consensus engine, which are \textit{Propose(p)} and \textit{Commit(p)}.
The leader replica uses \textit{Propose(p)} to broadcast a new proposal $p$ and \textit{Commit(p)} to commit $p$ when the order of $p$ is agreed on across the replicas (i.e., total ordering).
As illustrated in Figure~\ref{fig:architecture}, the transaction processing in state machine replication using SMP consists of the following steps:
\begin{itemize}[leftmargin=*]
    \item \textcircled{1} Upon receiving a new transaction $tx$ from the network, a replica calls \textit{ReceiveTx(tx)} to add $tx$ into the mempool, and \textcircled{2} disseminates $tx$ by calling \textit{ShareTx(tx)} if $tx$ is from a client (avoiding re-sharing if $tx$ is from a replica).
    \item \textcircled{3} Once the replica becomes the leader, it obtains a proposal (with transaction ids) $p$ by calling \textit{MakeProposal()}, and \textcircled{4} proposes it via \textit{Propose(p)}.
    \item \textcircled{5} Upon receipt of a proposal $p$, a non-leader replica calls \textit{FillProposal(p)} to reconstruct $p$ (pulling referenced transaction from the mempool), which is sent to the consensus engine to continue the consensus process.
    \item \textcircled{6} The consensus engine calls $Commit(p)$ to send committed proposals to the executor for execution.
\end{itemize}

\subsection{Data Structure}

\textit{\textbf{Microblock.}}
Transactions are collected from clients and batched into microblocks for dissemination\footnote{We use \emph{microblocks} and \emph{transactions} interchangeably throughout the paper.
For example, the \textit{ShareTx(tx)} primitive broadcasts a microblock instead of a single transaction in practice.}.
This is to amortize the verification cost.
Recall that we assume a client only sends a request to a single replica, which makes the microblocks sent from a replica disjoint from others.
Each microblock has a unique id calculated from the transaction ids it contains.

\begin{figure}[t]
    \centering
    \includegraphics[width=0.45\textwidth]{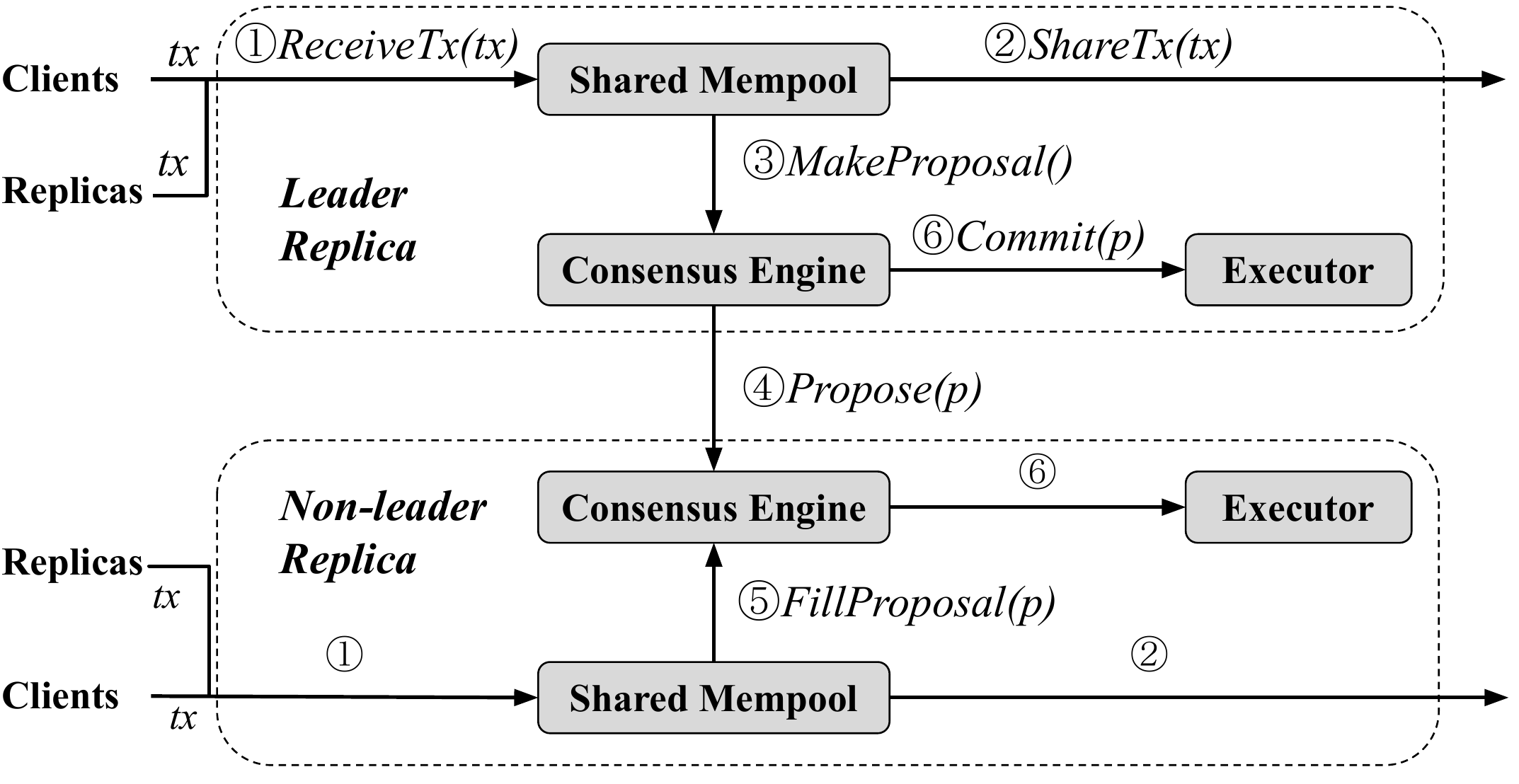}
    \caption{The processing of transactions in state machine replication using SMP.}
    \label{fig:architecture}
\end{figure}

\textit{\textbf{Proposal.}}
The \textit{MakeProposal()} primitive generates a proposal that consists of an id list of the microblocks and some metadata (e.g., the hash of the previous block, root hash of the microblocks).

\textit{\textbf{Block.}}
A block is obtained by calling the \textit{FillProposal(p)} primitive.
If all the microblocks referenced in a proposal $p$ can be found in the local mempool, we call it a \textit{full block}, or a \textit{full proposal}.
Otherwise, we call it a \textit{partial block/proposal}.
A block contains all the data included in the relevant proposal and a list of microblocks.

\subsection{Challenges and Solutions}\label{sec:challenges}

Here we discuss two challenges and corresponding solutions in implementing our SMP protocol.

\vspace{1mm}
\noindent \textbf{Problem-I: missing transactions lead to bottlenecks.}
\fangyu{Using best-effort broadcast~\cite{cachin2011introduction} to implement \textit{ShareTx(tx)} cannot ensure \textbf{SMP-Stability} since some referenced transactions (i.e., microblocks) in a proposal might never be received due to Byzantine behavior~\cite{leopard}.
Even in a Byzantine-free case, it is possible that a proposal arrives earlier than some of the referenced transactions.
We call these transactions \textit{missing transactions}.}
Figure~\ref{fig:mempool-set} illustrates an example in which a Byzantine broadcaster ($R_5$) only shares a transaction ($tx_{1}$) with the leader ($R_1$), not the other replicas.
Therefore, when $R1$ includes $tx_{1}$ in a proposal, $tx_{1}$ will be missing at the receiving replicas.
On the one hand, missing transactions block the consensus instance because the integrity of a proposal depends on the availability of the referenced transactions, which is essential to the security of a blockchain.
This could cause frequent view-changes which significantly affect performance, as we will show in Section~\ref{sec:eval-byz}.
On the other hand, to ensure \textbf{SMP-Stability}, replicas have to proactively fetch missing transactions from the leader.
This, however, creates a new bottleneck.
\fangyu{It is also difficult for the leader to distinguish between legitimate and malicious transaction requests.}

A natural solution to address the above challenge is to use reliable broadcast (RB)~\cite{narwal} to implement \textit{ShareTx(tx)}.
However, Byzantine reliable broadcast has quadratic message complexity and needs three communication rounds (round trip delay)~\cite{cachin2011introduction}, which is not suitable for large-scale systems.
We observe that some properties of reliable broadcast are \emph{not} needed by SMP since they can be provided by the consensus protocol itself (i.e., consistency and totality).
This enlightens us to seek for a lighter broadcast primitive.

\vspace{1mm}
\noindent \textbf{Solution-I: provably available broadcast.}
We resolve this problem by introducing a \textit{provably available broadcast} (PAB) primitive to ensure the availability of transactions referenced in a proposal with negligible overhead.
\fangyu{PAB provides an API to generate an \textit{availability proof} with at least $f+1$ signatures. 
Since at most $f$ signatures are Byzantine, the \textit{availability proof} guarantees that at least one correct replica (excluding the sender) has the message. This guarantees that the message can be eventually fetched from at least one correct replica.}
As such, by using PAB in Stratus, if a proposal contains valid available proofs for each referenced transaction, it can be passed directly to the commit phase without waiting for the transaction contents to arrive.
Therefore, missing transactions can be fetched using background bandwidth without blocking the consensus.

\vspace{1mm}
\noindent \textbf{Problem-II: unbalanced workload/bandwidth distribution.}
In deploying a BFT system across datacenters, it is difficult to ensure that all the nodes have identical resources.
Even if all the nodes have similar resources, it is unrealistic to assume that they will have a balanced workload in time and space.
This is because clients are unevenly distributed across regions and tend to use a preferred replica (nearest or most trusted).
In these cases, replicas with a low ratio of workload to bandwidth become bottlenecks.

To address the heterogeneity in workload/bandwidth, 
one popular approach is gossip~\cite{hyperledger-gossip,tendermint-srds, algorand}: the broadcaster randomly picks some of its peers and sends them the message, and the receivers repeat this process until all the nodes receive the message with high probability.
Despite their scalability, gossip protocols have a long tail-latency (the time required for the last node to receive the message) and high redundancy.

\begin{figure}[t]
    \centering
    \includegraphics[width=0.42\textwidth]{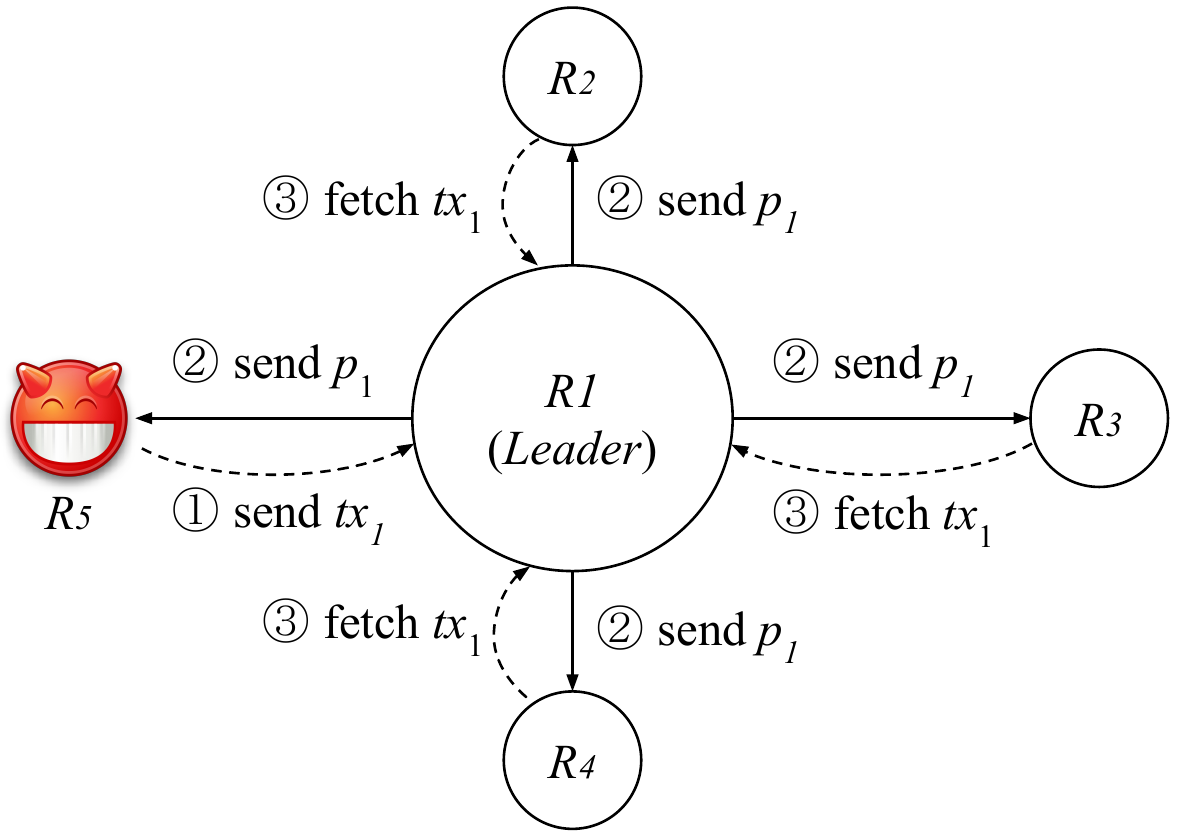}
    \caption{In a system with SMP, consisting of 5 replicas in which $R_5$ is Byzantine and $R_1$ is the current leader.}
    \label{fig:mempool-set}
\end{figure}

\vspace{1mm}
\noindent \textbf{Solution-II: distributed load balancing.} 
We address the challenge by introducing a distributed load-balancing (DLB) protocol that is co-designed with PAB.
DLB works locally at each replica and dynamically estimates a replica's local workloads and capacities so that overloaded replicas can forward their excess load (microblocks) to under-utilized replicas (proxies).
A proxy can disseminate a certain microblock on behalf of the original sender and prove that a microblock is successfully distributed by submitting available proof to the sender.
If the proof is not submitted in time, the sender picks another under-utilized replica and repeats the process.



\section{Transaction Dissemination}~\label{sec:dissemination}

We now introduce a new broadcast primitive called \textit{provably available broadcast} (PAB) for transaction dissemination, which mitigates the impact of missing transactions (\textbf{Problem-I}).
\fangyu{
\textit{Every replica} in Stratus runs PAB to distribute microblocks and collect availability proofs (threshold signatures).
When a replica becomes the leader, it pulls microblock ids as well as corresponding proofs into a proposal.
This ensures that every receiving replica will have an availability proof for all the referenced microblocks in a valid proposal.
These proofs resolve Problem I (Section~\ref{sec:challenges}) by providing \textbf{PAB-Provable Availability}. This ensures that a replica will eventually receive all the referenced microblocks and it does not need to wait for missing microblocks to arrive.}

\fangyu{
Broadcasting microblocks and collecting proofs is a distributed process that is not on the critical path of consensus. As a result, they will \emph{not} increase latency. In fact, we found that PAB significantly improves throughput and latency (Figure~\ref{fig:scalability}).
}


\subsection{Provably Available Broadcast}\label{sec:pab-instance}
In PAB, the sending replica, or sender, $s$ broadcasts a message $m$, collects acknowledgements of receiving the message $m$ from other replicas, and produces a succinct proof $\sigma$ (realized via threshold signature~\cite{threshold}) over $m$, showing that $m$ is available to at least one correct replica, say $r$.
Eventually, other replicas that do not receive $m$ from $s$ retrieves $m$ from $r$.
Formally, PAB satisfies the following properties:

\textbf{PAB-Integrity:} If a correct replica delivers a message $m$ from sender $s$, and $s$ is correct, then $m$ was previously broadcast by $s$.

\textbf{PAB-Validity:} If a correct sender broadcasts a message $m$, then every correct replica eventually delivers $m$.

\textbf{PAB-Provable Availability:} If a correct replica $r$ receives a valid proof $\sigma$ over $m$, then $r$ eventually delivers $m$.

We divide the algorithm into two phases, the \textit{push} phase and the \textit{recovery} phase.
The communication pattern is illustrated in Figure~\ref{fig:PAB}. \fangyu{We use angle brackets to denote messages and events and assume that messages are signed by their senders.}
In the \textit{push} phase, the \textit{sender} broadcasts a message $m$ and each receiver (including the sender) sends a \texttt{PAB-Ack} message $\left \langle \texttt{PAB-Ack}|m.id \right \rangle$ back to the \textit{sender}.
As long as the \textit{sender} receives at least a quorum of $q=f+1$ \texttt{PAB-Ack} messages (including the sender) from distinct receivers, it produces a succinct proof $\sigma$ (realized via threshold signature), showing that $m$ has been delivered by at least one correct replica.
The \textit{recovery} phase begins right after $\sigma$ is generated, and the \textit{sender} broadcasts the proof message $\left \langle \texttt{PAB-Proof}|id,\sigma \right \rangle$.
If some replica $r$ receives a valid \texttt{PAB-Proof} without receiving $m$, $r$ fetches $m$ from other replicas in a repeated manner.

\begin{figure}[t]
    \centering
    \includegraphics[width=0.43\textwidth]{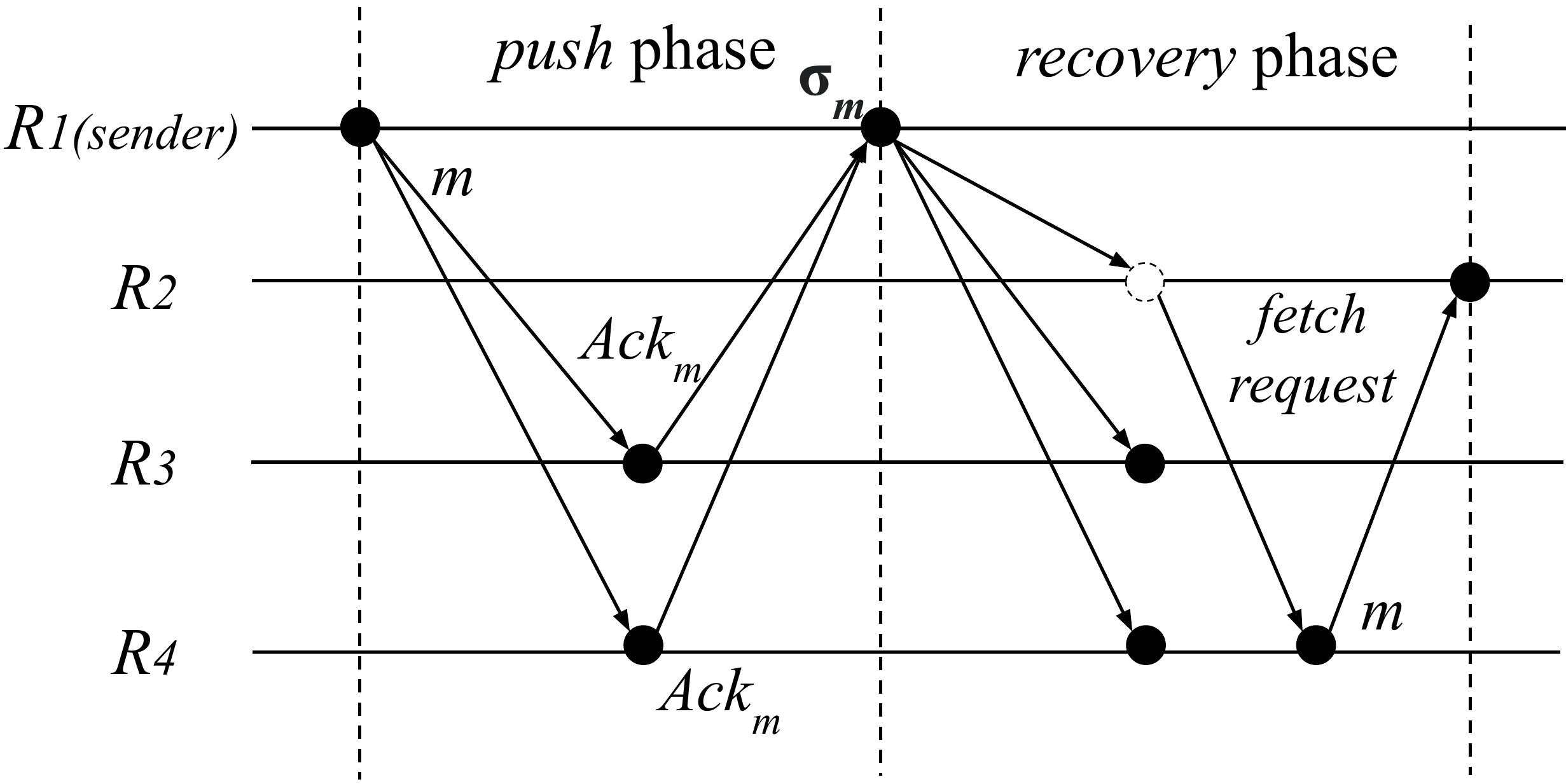}
    \caption{Message flow in PAB with $N=4$ replicas and $f=1$.
    $R_1$ is the sender (Byzantine).
    $R_2$ did not receive $m$ in the \textit{push} phase because of $R1$ or network asynchrony.
    Thus, $R_2$ fetches $m$ from $R_4$ (randomly picked) in the \textit{recovery} phase.}
    \label{fig:PAB}
\end{figure}

\begin{algorithm}[t]
  \footnotesize
  \caption{PAB with message $m$ at $R_i$ (\textit{push} phase)}\label{alg:dissemination}
  \begin{algorithmic}[1]
  \State\textbf{Local Variables}:
  \State $S \leftarrow\{\} $ \textcolor{gray}{\Comment{signature set over $m.id$}}
  \State $q \leftarrow f+1$ \textcolor{gray}{\Comment{quorum value adjustable between [$f+1$, $2f+1$]}}
  \Statex
  \State
  \textbf{upon event} $\left \langle \textsc{PAB-Broadcast}|m \right \rangle$ \textbf{do} \textsf{Broadcast}$(\left \langle \texttt{PAB-Msg}|m,R_i \right \rangle)$
  \Statex
  \State
  \textbf{upon receipt} $\left \langle \texttt{PAB-Msg}|m,s \right \rangle$ for the first time \textbf{do} 
  \textcolor{gray}{\Comment{$s\in C \cup R$}}
  \State\hspace{\algorithmicindent}
  \textsf{Store}$(m)$ \textcolor{gray}{\Comment{for future request}}
  \State\hspace{\algorithmicindent}
  \textbf{trigger} $\left \langle \textsc{PAB-Deliver}|m \right \rangle$
  \State\hspace{\algorithmicindent}
  \textbf{if} $s\in C$ \textbf{then trigger} $\left \langle \textsc{PAB-Broadcast}|m \right \rangle$\label{line:broadcast}
  \State\hspace{\algorithmicindent}
  \textbf{else} \textsf{Send}$(s,\left \langle \texttt{PAB-Ack}|m.id, R_i \right \rangle)$
  \Statex
  \State
  \textbf{upon receipt} $\left \langle \texttt{PAB-Ack}|id,s_j \right \rangle$ \textbf{do}\label{line:handle-ack}\textcolor{gray}{\Comment{if $R_i$ is the sender}}
  \State\hspace{\algorithmicindent} 
  $S\leftarrow S\cup s_j$
  \State\hspace{\algorithmicindent}
  \textbf{if} $|S|\geq q$ \textbf{then} \textcolor{gray}{\Comment{satisfies the quorum condition}}
  \State\hspace{\algorithmicindent}\hspace{\algorithmicindent}
  $\sigma\leftarrow\textsf{threshold-sign}(S)$
  \State\hspace{\algorithmicindent}\hspace{\algorithmicindent}
  \textbf{trigger} $\left \langle \textsc{PAB-Ava}|id, \sigma \right \rangle$
  \end{algorithmic}
\end{algorithm}

Algorithm~\ref{alg:dissemination} shows the \textit{push} phase, which consists of two rounds of message exchanges.
In the first round, the broadcaster disseminates $m$ via \textsf{Broadcast}() when the \textsc{PAB-Broadcast} event is triggered.
Note that a replica triggers \textsc{PAB-Broadcast} only if $m$ is received from a client to avoid re-sharing (Line~\ref{line:broadcast}). \fangyu{We use $C$ to denote the client set and $R$ to denote the replica set.} 
In the second round, every replica that receives $m$ acts as a witness by sending the sender a \texttt{PAB-Ack} message over $m.id$ (including the signature).
If the sender receives at least $q$ \texttt{PAB-Ack} messages for $m$ from distinct replicas, it generates a proof $\sigma$ from associated signatures via \textsf{threshold-sign}$()$ and triggers a \textsc{PAB-Ava} event.
The value of $q$ will be introduced shortly.

The \textit{recovery} phase serves as a backup in case the Byzantine senders only send messages to a subset of replicas or if messages are delayed due to network asynchrony.
The pseudocode of the \textit{recovery} phase is presented in Algorithm~\ref{alg:recovery}.
The sender broadcasts the proof $\sigma$ of $m$ on event \textsc{PAB-Ava}.
After verifying $\sigma$, the replica that has not received the content of $m$ invokes the $\textsf{PAB-Fetch}()$ procedure, which sends \texttt{PAB-Request} messages to a subset of replicas that are randomly picked from $signers$ of $\sigma$ (excluding replicas that have been requested).
The function $\textsf{random([0,1])}$ returns a random real number between $0$ and $1$. 
The configurable parameter $\alpha$ denotes the probability that a replica is requested. 
If the message is not fetched in $\delta$ time, the $\textsf{PAB-Fetch}()$ procedure will be invoked again and the timer will be reset. 

Although we use $q=f+1$ as the stability parameter in the previous description of PAB, the threshold is adjustable between $f+1$ and $2f+1$ without hurting PAB's properties.
The upper bound is $2f+1$ because there are $N\geq3f+1$ replicas in total, where up to $f$ of them are Byzantine.
In fact, $q$ captures a trade-off between the efficiency of the \textit{push} and \textit{recovery} phases.
A larger $q$ value improves the \textit{recovery} phase since it increases the chance of fetching the message from a correct replica. But, a larger $q$ increases latency, since it requires that the replica waits for more acks in the \textit{push} phase.

\begin{algorithm}[t]
  \footnotesize
  \caption{PAB with message $m$ at $R_i$ (\textit{recovery} phase)}\label{alg:recovery}
  \begin{algorithmic}[1]
  \State\textbf{Local Variables}:
  \State $signers \leftarrow\{\} $ \textcolor{gray}{\Comment{signers of $m$}}
  \State $requested \leftarrow\{\} $ \textcolor{gray}{\Comment{replicas that have been requested}}
  \Statex
  \State
  \textbf{upon event} $\left \langle \textsc{PAB-Ava}|id,\sigma \right \rangle$ \textbf{do}\textcolor{gray}{\Comment{if $R_i$ is the sender}}\label{line:fetch}
  \State\hspace{\algorithmicindent}
  \textsf{Broadcast}$(\left \langle \texttt{PAB-Proof}|id,\sigma \right \rangle)$
  \Statex
  \State
  \textbf{upon receipt} $\left \langle \texttt{PAB-Proof}|id,\sigma \right \rangle$ \textbf{do}
  \State\hspace{\algorithmicindent}
  \textbf{if} $\textsf{threshold-verify}(id,\sigma)$ is not \textbf{true} \textbf{do} \textbf{return}
  \State\hspace{\algorithmicindent}
  $signers\leftarrow\sigma.signers$
  \State\hspace{\algorithmicindent}
  \textbf{if} $m$ does not exist by checking $id$ \textbf{do} \textsf{PAB-Fetch}$(id)$
  \Statex
  \State
  \textbf{procedure} \textsf{PAB-Fetch}($id$)
  \State\hspace{\algorithmicindent}
  $\textsf{starttimer}(\texttt{Fetch},\delta,id)$\label{line:starttimer}
  \State\hspace{\algorithmicindent}
  \textbf{forall} $r\in signers\setminus requested$ \textbf{do}
  \State\hspace{\algorithmicindent}\hspace{\algorithmicindent}
  \textbf{if} $\textsf{random}([0,1])>\alpha$ \textbf{then}
  \State\hspace{\algorithmicindent}\hspace{\algorithmicindent}\hspace{\algorithmicindent}
  $requested\leftarrow requested \cup r$
  \State\hspace{\algorithmicindent}\hspace{\algorithmicindent}\hspace{\algorithmicindent}
  $\textsf{Send}(r, \left \langle \texttt{PAB-Request}|id,R_i \right \rangle)$
  \State\hspace{\algorithmicindent}
  \textbf{wait until} all requested messages are delivered, or $\delta$ timeout \textbf{do}
  \State\hspace{\algorithmicindent}\hspace{\algorithmicindent}
  \textbf{if} $\delta$ timeout \textbf{do} \texttt{PAB-Fetch}$(id)$
  \end{algorithmic}
\end{algorithm}

\subsection{Using PAB in Stratus}\label{sec:stable-consensus}

Now we discuss how we use PAB in our Stratus Mempool and how it is integrated with a leader-based BFT protocol.
Recall Figure~\ref{fig:architecture} that shows the interactions between the shared mempool and the consensus engine in the \textit{Propose} phase.
Specifically, (i) the leader makes a proposal by calling \textsf{MakeProposal}(), and (ii) upon a replica receiving a new proposal $p$, it fills $p$ by calling \textsf{FillProposal}$(p)$.
Here we present the implementations of the \textsf{MakeProposal()} and \textsf{FillProposal}$(p)$ procedures as well as the logic for handling an incoming proposal in Algorithm~\ref{alg:proposal}.
The consensus events and messages are denoted with \textsc{CE}.

\begin{algorithm}[t!]
  \footnotesize
  \caption{\textit{Propose} phase of view $v$ at replica $R_i$}\label{alg:proposal}
  \begin{algorithmic}[1]
  \State\textbf{Local Variables}:
  \State $mbMap \leftarrow\{\} $ \textcolor{gray}{\Comment{maps microblock id to microblock}}
  \State $pMap \leftarrow\{\} $ \textcolor{gray}{\Comment{maps microblock id to available proof}}
  \State $avaQue\leftarrow\{\}$ \textcolor{gray}{\Comment{stores microblock id that is provably available}}
  \Statex
  \State
  \textbf{upon receipt} $\left \langle \texttt{PAB-Proof}|id,\sigma \right \rangle$ \textbf{do}
  \State\hspace{\algorithmicindent}
  \textbf{if} $\textsf{threshold-verify}(id,\sigma)$ is not \textbf{true} \textbf{do} \textbf{return}
  \State\hspace{\algorithmicindent}
  $pMap[id]\leftarrow \sigma$\label{line:store-proof}
  \State\hspace{\algorithmicindent}
  $avaQue.\textsf{Push}(id)$\label{line:store-ava}
  \Statex
  \State
  \textbf{upon event} $\left \langle \textsc{PAB-Deliver}|mb \right \rangle$ \textbf{do} $mbMap[mb.id]\leftarrow mb$\label{line:deliver}
  \Statex
  \State
  \textbf{upon event} $\left \langle \textsc{CE-NewView}|v \right \rangle$ \textbf{do}
  \State\hspace{\algorithmicindent}
  \textbf{if} $R_i$ is the leader for view \textit{v} \textbf{then}
  \State\hspace{\algorithmicindent}\hspace{\algorithmicindent}
  $p\leftarrow\textsf{MakeProposal}(v)$
  \State\hspace{\algorithmicindent}\hspace{\algorithmicindent}
  $\textsf{Broadcast}(\left \langle \texttt{CE-Propose}|p,R_i\right \rangle)$
  \Statex
  \State
  \textbf{procedure} \textsf{MakeProposal}$(v)$
  \State\hspace{\algorithmicindent}
  $payload\leftarrow\{\}$
  \State\hspace{\algorithmicindent}
  \textbf{while(\textbf{True})}
  \State\hspace{\algorithmicindent}\hspace{\algorithmicindent}
  $id\leftarrow avaQue.\textsf{Pop}()$
  \State\hspace{\algorithmicindent}\hspace{\algorithmicindent}
  $payload[id]\leftarrow pMap[id]$
  \State\hspace{\algorithmicindent}\hspace{\algorithmicindent}
  \textbf{if} $\textsf{Len}(payload)\geq\textsc{BlockSize}$ \textbf{or} $id=\perp$ \textbf{then}
  \State\hspace{\algorithmicindent}\hspace{\algorithmicindent}\hspace{\algorithmicindent}
  \textbf{break}\label{line:pop-end}
  \State\hspace{\algorithmicindent}
  \textbf{return} \textsf{newProposal($v,payload$)}
  \Statex
  \State
  \textbf{upon receipt} $\left \langle \texttt{CE-Propose}|p,r \right \rangle$ \textcolor{gray}{\Comment{$r$ is the current leader}}
  \State\hspace{\algorithmicindent}
  \textbf{for} $id, \sigma\in p.payload$ \textbf{do}
  \Statex\hspace{\algorithmicindent}\hspace{\algorithmicindent}
  \textbf{if} $\textsf{threshold-verify}(id,\sigma)$ is not \textbf{true} \textbf{do}
  \State\hspace{\algorithmicindent}\hspace{\algorithmicindent}\hspace{\algorithmicindent}
  \textbf{trigger} $\left \langle \textsc{CE-ViewChange}|R_j \right \rangle$
  \State\hspace{\algorithmicindent}\hspace{\algorithmicindent}\hspace{\algorithmicindent}
  \textbf{return}
  \State\hspace{\algorithmicindent}
  \textbf{trigger} $\left \langle \textsc{CE-EnterCommit}|p \right \rangle$\label{line:enter-commit}
  \State\hspace{\algorithmicindent}
  \textsf{FillProposal}$(p)$
  \Statex
  \State
  \textbf{procedure} \textsf{FillProposal}$(p)$
  \State\hspace{\algorithmicindent}
  $block\leftarrow \{p\}$
  \State\hspace{\algorithmicindent}
  \textbf{forall} $id\in p.payload$ \textbf{do}
  \State\hspace{\algorithmicindent}\hspace{\algorithmicindent}
  \textbf{if} $mb$ associated with $id$ has not been delivered \textbf{then}
  \State\hspace{\algorithmicindent}\hspace{\algorithmicindent}\hspace{\algorithmicindent}
  $\textsf{PAB-Fetch}(id)$
  \State\hspace{\algorithmicindent}
  \textbf{wait until} every requested $mb$ is delivered \textbf{then}
  \State\hspace{\algorithmicindent}\hspace{\algorithmicindent}
  \textbf{forall} $id\in p.payload$ \textbf{do}
  \State\hspace{\algorithmicindent}\hspace{\algorithmicindent}\hspace{\algorithmicindent}
  $block.\textsf{Append}(mbMap[id])$
  \State\hspace{\algorithmicindent}\hspace{\algorithmicindent}\hspace{\algorithmicindent}
  $avaQue.\textsf{Remove}(id)$\label{line:remove}
  \State\hspace{\algorithmicindent}\hspace{\algorithmicindent}\hspace{\algorithmicindent}
  \textbf{trigger} $\left \langle \textsc{CE-Full}|block \right \rangle$
  \end{algorithmic}
\end{algorithm}

Since transactions are batched into microblocks for dissemination,
we use microblocks (i.e., $mb$) instead of transactions in our description.
The consensus protocol subscribes \textsc{PAB-Deliver} events and \texttt{PAB-Proof} messages from the underlying PAB protocol and modifies the handlers, in which we use \textit{mbMap}, \textit{pMap}, and \textit{avaQue} for bookkeeping.
Specifically, \textit{mbMap} stores microblocks upon the \textsc{PAB-Deliver} event (Line~\ref{line:deliver}).
Upon the receipt of \texttt{PAB-Proof} messages, the microblock \textit{id} is pushed into the queue \textit{avaQue} (Line~\ref{line:store-ava}) and the relevant proof $\sigma$ is recorded in \textit{pMap} (Line~\ref{line:store-proof}).

We assume the consensus protocol proceeds in views, and each view has a designated leader.
A new view is initiated by a \textsc{CE-NewView} event.
Once a replica becomes the leader for the current view, it attempts to invoke the \textsf{MakeProposal()} procedure, which pulls microblocks (only ids) from the front of $avaQue$ and piggybacks associated proofs.
It stops pulling when the number of contained microblocks has reached \textsc{BlockSize}, or there are no microblocks left in $avaQue$.
The reason why the proposal needs to include all the associated available proofs of each referenced transaction is to show that the availability of each referenced microblock is guaranteed.
We argue that the inevitable overhead is negligible provided that the microblock is large.

On the receipt of an incoming proposal $p$, the replica verifies every proof included in $p.payload$ and triggers a \textsc{CE-ViewChange} event if the verification is not passed, attempting to replace the current leader.
If the verification is passed, a $\left \langle \textsc{CE-EnterCommit}|p \right \rangle$ event is triggered and the processing of $p$ enters the commit phase (Line~\ref{line:enter-commit}).
Next, the replica invokes the \textsf{FillProposal}$(p)$ procedure to pull the content of microblocks associated with $p.payload$ from the mempool.
The \textsf{PAB-Fetch}$(id)$ procedure (Algorithm~\ref{alg:recovery}) is invoked when missing microblocks are found.
The thread waits until all the requested microblocks are delivered.
Note that this thread is independent of the thread handling consensus events.
Therefore, waiting for requested microblocks will not block consensus.
After a \textit{full block} is constructed, the replica triggers a $\left \langle \textsc{CE-Full}|block \right \rangle$ event, indicating that the block is ready for execution.

\fangyu{In Stratus, the transactions in a microblock are executed if and only if all transactions in the previous microblocks are received and executed.}
Since missing transactions are fetched according to their unique ids, consistency is ensured.
\fangyu{
Therefore, using Stratus in any case will not compromise the safety of the consensus protocol.
}
\textbf{The advantage of using PAB is that it allows the consensus protocol to safely enter the commit phase of a proposal without waiting for the missing microblocks to be received.}
In addition, the \textit{recovery} phase proceeds concurrently with the consensus protocol (only background bandwidth is used) until the associated block is full for execution.
Many optimizations~\cite{ace,execution,basil} for improving the execution have been proposed and we hope to build on them in our future work.
Our implementation satisfies \textbf{PAB-Provable Availability}, which helps Stratus achieve \textbf{SMP-Inclusion} and \textbf{SMP-Stability}.

\subsection{Correctness Analysis}

Now we prove the correctness of PAB.
Since the integrity and validity properties are simple to prove, here we only show that Algorithm~\ref{alg:dissemination} and Algorithm~\ref{alg:recovery} satisfy \textbf{PAB-Provable Avalability}.
Then we provide proofs that Stratus satisfies \textbf{SMP-Inclusion} and \textbf{SMP-Stability}.

\begin{lemma}
[\textbf{PAB-Provable Availability}] If a proof $\sigma$ over a message $m$ is valid, then at least one correct replica holds $m$.
In the \textit{recovery} phase (Algorithm~\ref{alg:recovery}), the receiving replica $r$ repeatedly invokes \textsf{PAB-Fetch}$(id)$ and sends requests to randomly picked replicas.
Eventually, a correct replica will respond and $r$ will deliver $m$.
\end{lemma}

\begin{theorem}
Stratus ensures \textbf{SMP-Inclusion}.
\end{theorem}

\begin{proof}
If a transaction $tx$ is delivered and verified by a correct replica $r$ (the sender), it will be eventually batched into a microblock $mb$ and disseminated by PAB.
Due to the \textit{validity} property of PAB, $mb$ will be eventually delivered by every correct replica, which sends acks over $mb$ back to the sender.
An available proof $\sigma$ over $mb$ will be generated and broadcast by the sender.
Upon the receipt of $\sigma$, every correct replica pushes $mb$ ($mb.id$) into $avaQue$.
Therefore, $mb$ ($tx$) will be eventually popped from $avaQue$ of a correct leader $l$ and proposed by $l$.
\end{proof}

\begin{theorem}
Stratus ensures \textbf{SMP-Stability}. 
\end{theorem}

\begin{proof}
If a transaction $tx$ is included in a proposal by a correct leader, it means that $tx$ is provably available (a valid proof $\sigma$ over $tx$ is valid).
Due to the \textbf{PAB-Provable Availability} property of PAB, every correct replica eventually delivers $tx$.
\end{proof}

\section{Load Balancing}\label{sec:load-balance}

We now discuss Stratus' load balancing.
Recall that replicas disseminate transactions in a distributed manner.
But, due to network heterogeneity and workload imbalance (\textbf{Problem-II}), performance will be bottlenecked by overloaded replicas.
Furthermore, a replica's workload and its resources may vary over time.
Therefore, a load balancing protocol that can adapt to a replica's workload and capacity is necessary.

In our design, busy replicas will forward excess load to less busy replicas that we term \textit{proxies}.
The challenges are (i) how to determine whether a replica is busy, (ii) how to decide which replica should receive excess loads, and (iii) how to deal with Byzantine proxies that refuse to disseminate the received load.

Our load balancing protocol works as follows.
A local workload estimator monitors the replica to determine if it is \textit{busy} or \textit{unbusy}.
We discuss work estimation in Section~\ref{subsec:estimation}.
Next, a \textit{busy} replica forwards newly generated microblocks to a proxy.
The proxy initiates a PAB instance with a forwarded microblock and is responsible for the \textit{push} phase.
When the \textit{push} phase completes, the proxy sends the \texttt{PAB-Proof} message of the microblock to the original replica, which continues the \textit{recovery} phase.
In addition, we adopt a $banList$ to avoid Byzantine proxies.
Next, we discuss how a busy replica forwards excess load.

\subsection{Load Forwarding}\label{sec:load-forwarding}

Before forwarding excess load, a busy replica needs to know which replicas are unbusy.
A na\"ive approach is to ask other replicas for their load status.
However, this requires all-to-all communications and is not scalable.
Instead, we use the well-known Power-of-d-choices (Pod) algorithm~\cite{vvedenskaya1996queueing, mitzenmacher2001power,batch}.
A \textit{busy} replica randomly samples load status from $d$ replicas, and forwards its excess load to the least loaded replica (the proxy).
Here, $d$ is usually much smaller than the number of replicas $N$.
Our evaluation shows that $d = 3$ is sufficient for a network with hundreds of nodes and unbalanced workloads (see Section~\ref{sec:eval-unbalance}).
Note that the choice of $d$ is independent of $f$; we discuss how we handle Byzantine proxies later in this section.
The randomness in Pod ensures that the same proxy is unlikely to be re-sampled and overloaded.

\begin{algorithm}[t!]
  \footnotesize
  \caption{The Load Forwarding procedure at replica $R_i$}\label{alg:balance}
  \begin{algorithmic}[1]
  \State\textbf{Local Variables}:
  \State $samples\leftarrow\{\}\{\}$ \textcolor{gray}{\Comment{stores sampled info for a microblock}}
  \State $banList\leftarrow\{\}$ \textcolor{gray}{\Comment{stores potentially Byzantine proxies}}
  \Statex
  \State
  \textbf{upon event} $\left \langle \textsc{NewMB}|mb \right \rangle$ \textbf{do} 
  \State\hspace{\algorithmicindent}
  \textbf{if} \textsf{IsBusy}() \textbf{do} \textsf{LB-ForwardLoad}$(mb)$
  \State\hspace{\algorithmicindent}
  \textbf{else trigger} $\left \langle \textsc{PAB-Broadcast}|mb \right \rangle$
  \Statex
  \State \textbf{procedure} $\textsf{LB-ForwardLoad}(mb)$ \textcolor{gray}{\Comment{if $R_i$ is the \textit{busy} sender}}
  \State\hspace{\algorithmicindent}
  $\textsf{starttimer}(\texttt{Sample},\tau,mb.id)$  
  \State\hspace{\algorithmicindent}
  $K \leftarrow \textsf{SampleTargets}(d)\setminus banList$
  \State\hspace{\algorithmicindent}
  \textbf{forall} $r\in K$ \textbf{do} $\textsf{Send}(r, \left \langle \label{line:sample} \texttt{LB-Query}|mb.id,R_i \right \rangle)$
  \State\hspace{\algorithmicindent}
  \textbf{wait until} $|samples[mb.id]|=d$ or $\tau$ timeout \textbf{do}
  \State\hspace{\algorithmicindent}\hspace{\algorithmicindent}
  \textbf{if} $|samples[mb.id]|=0$ \textbf{then}
  \State\hspace{\algorithmicindent}\hspace{\algorithmicindent}\hspace{\algorithmicindent}
  \textbf{trigger} $\left \langle \textsc{PAB-Broadcast}|mb \right \rangle$\label{line:no-replies}
  \State\hspace{\algorithmicindent}\hspace{\algorithmicindent}\hspace{\algorithmicindent}
  \textbf{return}
  \State\hspace{\algorithmicindent}\hspace{\algorithmicindent}
  \textbf{find} $r_{p}\in samples[mb.id]$ with the smallest $w$
  \State\hspace{\algorithmicindent}\hspace{\algorithmicindent}
  \textsf{starttimer}$(\texttt{Forward},\tau^{\prime},mb)$ \label{line:forward-timer}
  \State\hspace{\algorithmicindent}\hspace{\algorithmicindent}
  $banList.\textsf{Append}(r_p)$ \label{line:addbanlist}
  \fangyu{\Comment{every proxy is put in $banList$}}
  \State\hspace{\algorithmicindent}\hspace{\algorithmicindent}
  $\textsf{Send}(r_{p},\left \langle \texttt{LB-Forward}|mb, R_i\right \rangle)$ \textcolor{gray}{\Comment{send $mb$ to the poxy}}
  \State\hspace{\algorithmicindent}\hspace{\algorithmicindent}
  \textbf{wait until} \texttt{PAB-Proof} over $mb$ is received or $\tau^{\prime}$ timeout \textbf{do}
  \State\hspace{\algorithmicindent}\hspace{\algorithmicindent}\hspace{\algorithmicindent}
  \textbf{if} $\tau^{\prime}$ \textbf{do} \texttt{LB-ForwardLoad}($mb$)\label{line:resend}
  \State\hspace{\algorithmicindent}\hspace{\algorithmicindent}\hspace{\algorithmicindent}
  \fangyu{\textbf{else} $banList.\textsf{Remove}(R_p)$ \label{line:removebanlist}
  \hspace{-.35em}\Comment{$R_p$ is removed from $banList$}
  }
  \State \textbf{upon receipt} $\left \langle \texttt{LB-Forward}|mb,r \right \rangle$ \textbf{do} \textcolor{gray}{\Comment{if $R_i$ is the proxy with $mb$}}
  \State\hspace{\algorithmicindent}
  \textbf{trigger} $\left \langle \textsc{PAB-Broadcast}|mb \right \rangle$
  \Statex
  \State \textbf{upon receipt} $\left \langle \texttt{LB-Query}|id,r \right \rangle$ \textbf{do} \textcolor{gray}{\Comment{if $R_i$ is sampled}}
  \State\hspace{\algorithmicindent}
  $w\leftarrow \textsf{GetLoadStatus}()$
  \State\hspace{\algorithmicindent}
  $\textsf{Send}(r, \left \langle \texttt{LB-Info}|w, id, R_i \right \rangle)$ \label{line:pick-replica-end}
  \Statex
  \State \textbf{upon receipt} $\left \langle \texttt{LB-Info}|w, id, r \right \rangle$ \textbf{do} \textcolor{gray}{\Comment{if $R_i$ is \textit{busy}}}
  \State\hspace{\algorithmicindent}
  $samples[id][R_i]\leftarrow w$
  \Statex
  \State\textbf{upon receipt} $\left \langle \texttt{PAB-Proof}|id,\sigma \right \rangle$ before $\tau^{\prime}$ timeout \textbf{do}
  \State\hspace{\algorithmicindent}
  \textbf{if} $\textsf{threshold-verify}(id,\sigma)$ is not \textbf{true} \textbf{do} \textbf{return}
  \State\hspace{\algorithmicindent}
  \textbf{trigger} $\left \langle \textsc{PAB-Ava}|id, \sigma \right \rangle$ \textcolor{gray}{\Comment{$R_i$ takes over the \textit{recovery} phase}}
  
  \Statex
  \fangyu{\State\textbf{upon event} $\left \langle \textsc{Reset} | \textit{banList} \right \rangle$}
  \State\hspace{\algorithmicindent}
  $banList\leftarrow\{\}$
  \textcolor{gray}{\Comment{clear \textit{banList} periodically}}
  \label{line:continue-recovery}
  \end{algorithmic}
\end{algorithm}

Algorithm~\ref{alg:balance} depicts the \textsf{LB-ForwardLoad} procedure and relevant handlers. 
Upon the generation of a new microblock $mb$, the replica first checks whether it is \textit{busy} (see Section~\ref{subsec:estimation}).
If so, it invokes the \textsf{LB-ForwardLoad}$(mb)$ procedure to forward $mb$ to the proxy; otherwise, it broadcasts $mb$ using PAB by itself.
To select a proxy, a replica samples load status from $d$ random replicas (excluding itself) within a timeout of $\tau$ (Line~\ref{line:sample}).
Upon receiving a workload query, a replica obtains its current load status by calling the \textsf{GetLoadStatus}() (see Section~\ref{subsec:estimation}) and piggybacks it on the reply (Line 23-25). 
If the sender receives all the replies or times out, it picks the replica that replied with the smallest workload and sends $mb$ to it.
This proxy then initiates a PAB instance for $mb$ and sends the \texttt{PAB-Proof} message back to the original sender when a valid proof over $mb$ is generated.
Note that if no replies are received before timeout, the sending replica initiates a PAB instance by itself (Line~\ref{line:no-replies}).
Note that due to the decoupling design of Stratus, the overhead introduced by load forwarding has negligible impact on consensus.
\fangyu{To prevent a malicious replica from sending a small batch to reduce the performance, every replica can set a minimum batch size for receiving a batch.}

\fangyu{
In Stratus, each replica randomly and independently chooses $d$ replicas from the remaining $N$-1 replicas.
Since the workload of each replica changes quickly, \textbf{the sampling happens for each microblock} without blocking the forwarding process.
Therefore, for each load balancing event of an overloaded replica A, the probability that a specific replica (other than replica A) is chosen by replica A is $d/(N$-1). 
The probability that a replica is chosen by all replicas is very small. For example, when $d=3$ and $N=100$ the probability that a replica is chosen by more than 7 replicas is about $0.03$.
We omit the analysis details due to the page limit.
Next, we discuss how we handle Byzantine behaviors during load forwarding.
}

\textbf{Handling faults.} A sampled Byzantine replica can pretend to be unoccupied by responding with a low busy level and censoring the forwarded microblocks.
In this case, the \textbf{SMP-Inclusion} would be compromised: the transactions included in the censored microblock will not be proposed.
We address this issue as follows.
A replica $r$ sets a timer before sending $mb$ to a selected proxy $p$ (Line~\ref{line:forward-timer}).
If $r$ does not receive the available proof $\sigma$ over $mb$ before the timeout, $r$ re-transmits $mb$ by re-invoking the \textsf{LB-ForwardLoad}$(mb)$ (Line~\ref{line:resend}).
Here, the unique microblock ids prevent duplication.
The above procedure repeats until a valid $\sigma$ over $mb$ is received.
Then $r$ continues the \textit{recovery} phase of the PAB instance with $mb$ by triggering the \textsc{PAB-Ava} event (Line~\ref{line:continue-recovery}).

To prevent Byzantine replicas from being sampled again, we use a \textit{banList} to store proxies that have not finished the \textit{push} phase of a previous PAB instance.
\fangyu{That is, before a busy sender sends a microblock $mb$ to a proxy, the proxy is added to the \textit{banList} (Line~\ref{line:addbanlist}).
}
For future sampling, the replicas in the \textit{banList} are excluded.
\fangyu{As long as the sender receives a valid proof message for $mb$ from the proxy before a timeout, the proxy will be removed from the \textit{banList} (Line~\ref{line:removebanlist})}.
\fangyu{The \textit{banList} is periodically cleared by a timer to avoid replicas from being banned forever (Line~\ref{line:continue-recovery}).}
Note that more advanced \textit{banList} mechanisms can be used based on proxies' behavior~\cite{aardvark} and we consider to include them in our future work.

\subsection{Workload Estimation} \label{subsec:estimation}
Our workload estimator runs locally on an on-going basis and is responsible for estimating \textit{load status}.
Specifically, it determines: (i) whether the replica is overloaded, and (ii) how much the replica is overloaded, which correspond to the two functions in Algorithm~\ref{alg:balance}, \textsf{IsBusy}() and \textsf{GetLoadStatus}(), respectively.
To evaluate replicas' load status, two ingredients need to be considered: workload and capacity.
As well, the estimated results must be comparable across replicas in a heterogeneous network.

To address these challenges, we use \textit{stable time} (ST) to estimate a replica's load status.
The stable time of a microblock is measured from when the sender broadcasts the microblock until the time that the microblock becomes stable (receiving $f+1$ acks).
To estimate ST of a replica, the replica calculates the ST of each microblock if it is the sender and takes the $n$-th (e.g., $n=95$) percentile of the ST values in a window of the latest stable microblocks.
\arxiv{Figure~\ref{fig:workload-estimation} shows the estimation process.}
The estimated ST of a replica is updated when a new microblock becomes stable.
The window size is configurable and we use $100$ as the default size.

\begin{figure}[t]
    \centering
    \includegraphics[width=0.4\textwidth]{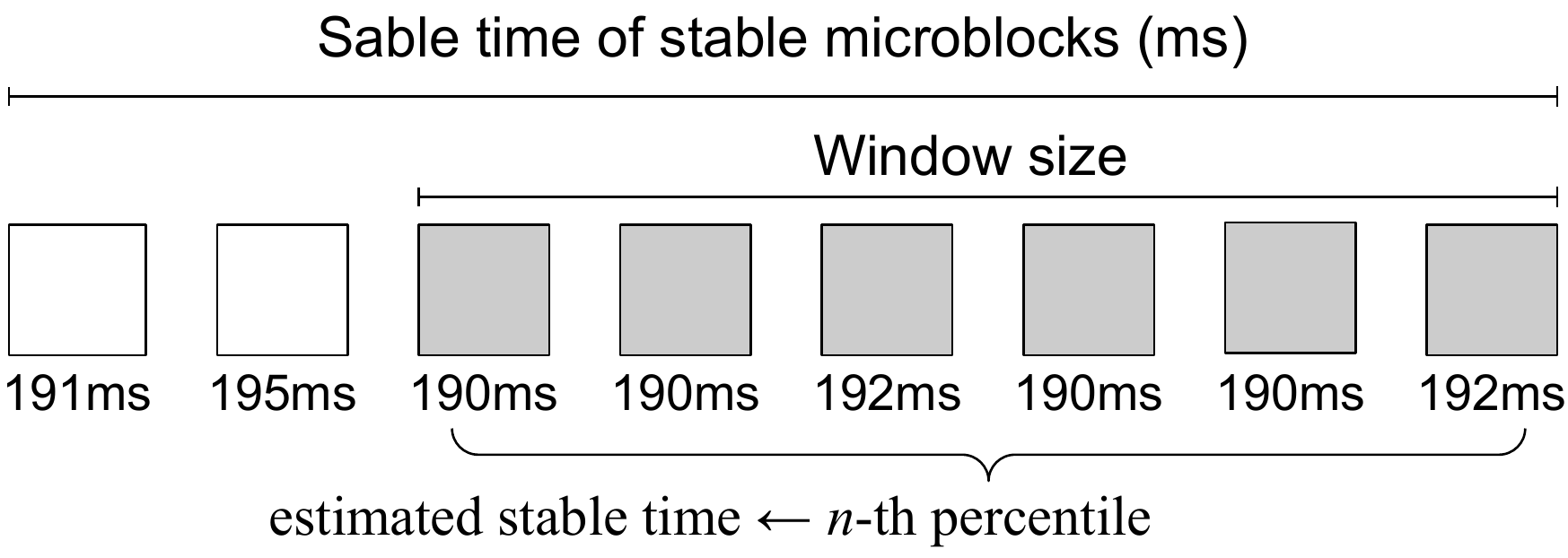}
    \caption{The stable time (ST) of a replica is estimated by taking the $n$-th percentile of ST values over a window of latest stable microblocks. The window slides when new microblocks become stable.}
    \label{fig:workload-estimation}
\end{figure}

\begin{figure}[t]
\centering
\begin{subfigure}{0.24\textwidth}
    \includegraphics[width=\textwidth]{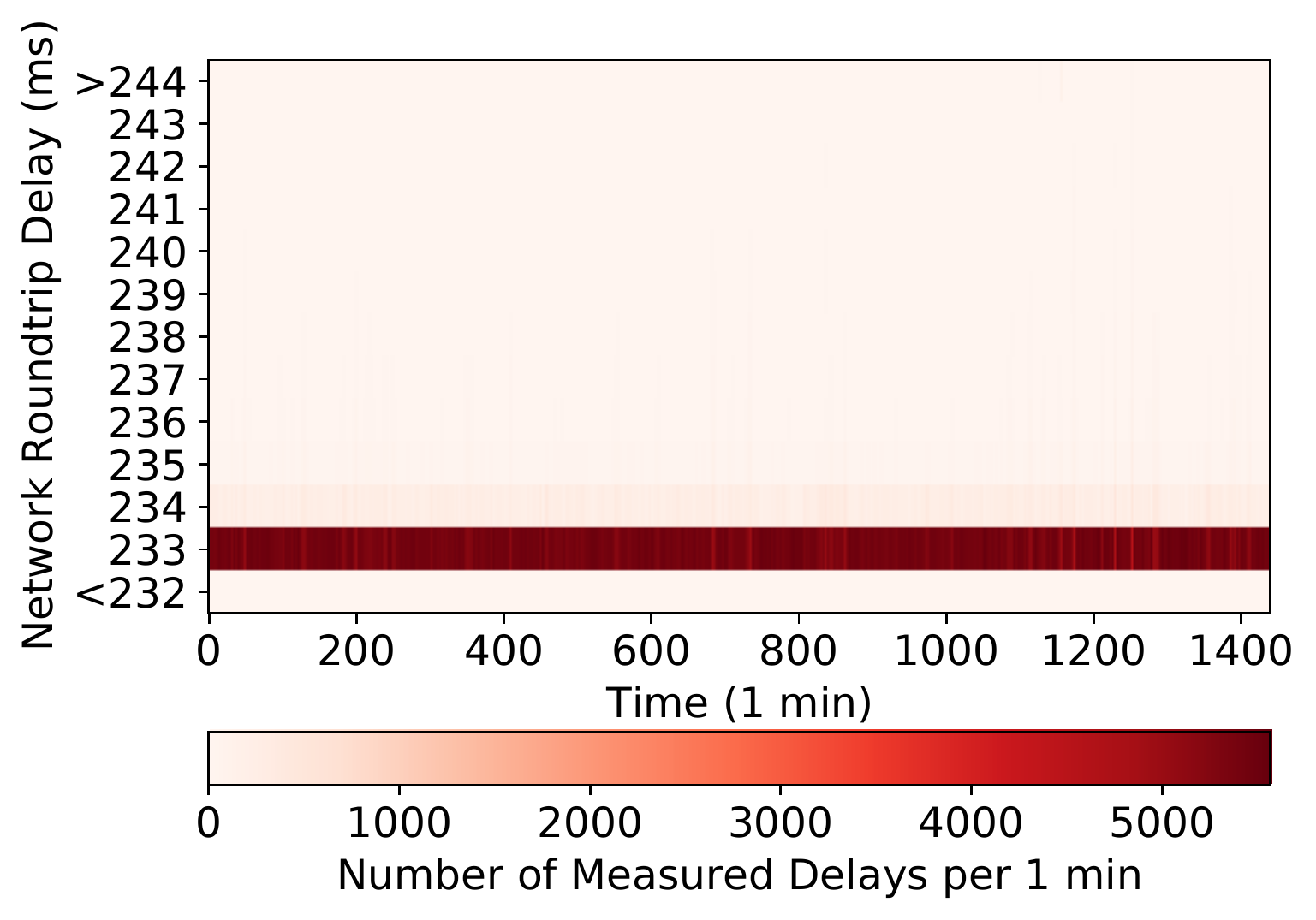}
    \caption{Heat map of measured roundtrip delays between servers from Virginia to Singapore over 24 hours.}
    \label{fig:heatmap}
\end{subfigure}
\hfill
\begin{subfigure}{0.24\textwidth}
    \includegraphics[width=\textwidth]{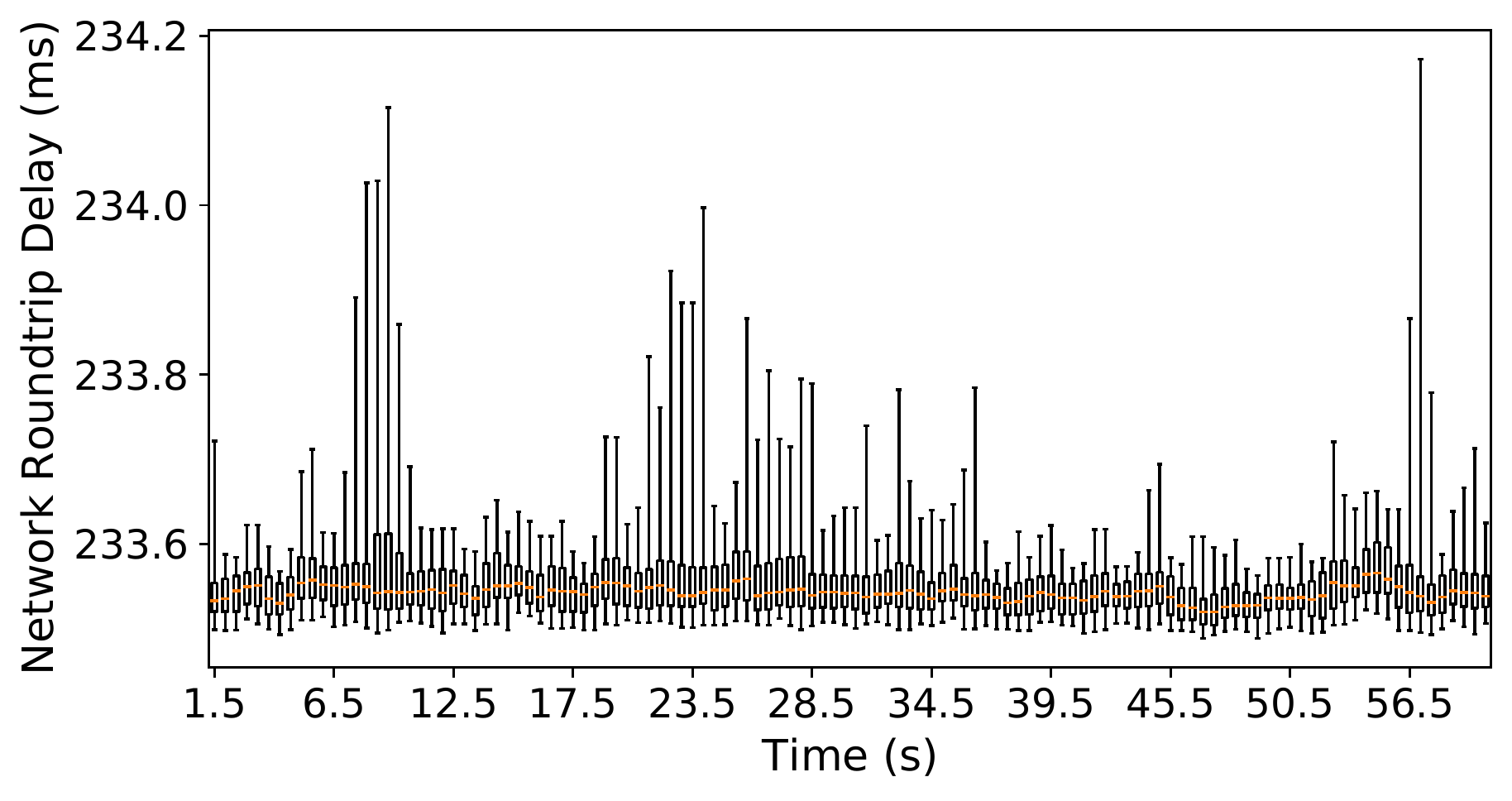}
    \caption{Distribution of measured delays between servers from Virginia to Singapore during 1 minute at 12th hour.}
    \label{fig:dist-delay}
\end{subfigure}
\caption{Network roundtrip delays between Virginia and Singapore.}
\label{fig:delay-measurements}
\end{figure}

Our approach is based on two observations.
First, the variability in network delay in a private network is small~\cite{domino}.
Second, network delay increases sharply when a node is overloaded.
\fangyu{The above observations are based on our measurements. A selection of these is shown in Figure~\ref{fig:delay-measurements}.
Figure~\ref{fig:heatmap} is a heat map of measured delays between two regions (Virginia and Singapore) in Alibaba Cloud over 24 hours.
Figure~\ref{fig:dist-delay} exhibits the round-trip delay distribution during 1 minute starting the 12th hour in the measurements.
We omit measurements of other pair of datacenters in this paper.
Our results demonstrate that the inter-datacenter network delays across different regions are stable and predictable based on recent measurement data.
}
Thus, under a constant workload, the calculated ST should be at around a constant number $\alpha$ with an error of $\epsilon$.
If the estimated ST is larger than $\alpha+\epsilon$ by a parameter of $\beta$, a replica is considered \textit{busy} (return \textit{true} in the \textsf{IsBusy} function).
Additionally, the value of ST reflects the degree to which a replica is loaded: the smaller the ST, the more resources a replica has for disseminating microblocks.
Therefore, we use the ST as the return value of the function \textsf{GetLoadStatus}.
Note that the \textsf{GetLoadStatus} returns a \textsc{NULL} value if the calling replica is \textit{busy}.
Also note that due to network topology, the ST value does not faithfully reflect the load status across replicas.
For example, some replicas may have a smaller ST because they are closer to a quorum of other replicas.
In this case, forwarding excess load to these replicas also benefits the system.
For overloaded replicas with large ST values, the topology has a negligible impact.
In case the network is unstable, we can also estimate the load status by monitoring the queue length of the network interface card.
We save that for future work.


\section{Implementation} \label{sec:implementation}

We prototyped Stratus\footnote{Available at https://github.com/gitferry/bamboo-stratus} in Go with Bamboo~\cite{bamboo}\footnote{Available at https://github.com/gitferry/bamboo}, which is an open source project for prototyping, evaluating, and benchmarking BFT protocols.
Bamboo provides validated implementations of state-of-the-art BFT protocols such as PBFT~\cite{pbft}, HotStuff~\cite{hotstuff}, and Streamlet~\cite{streamlet}.
Bamboo also supplies common functionalities that a BFT replication protocol needs.
In our implementation, we replaced the mempool in Bamboo with Stratus shared mempool.
Because of Stratus' well-designed interfaces, the consensus core is minimally modified.
\fangyu{We used HotStuff's Pacemaker for view change, though Stratus is agnostic to the view-change mechanism.}
Similar to~\cite{hotstuff,s-dumbo}, we use ECDSA to implement the quorum proofs in PAB instead of threshold signature.
This is because the computation efficiency of ECDSA\footnote{We trivially concatenate f + 1 ECDSA signatures} is better than Boldyreva’s threshold signature~\cite{s-dumbo}.
Overall, our implementation added about 1,300 lines of Go to Bamboo.

\vspace{1mm}
\noindent \textbf{Optimizations.}
Since microblocks consume the most bandwidth, we need to reserve sufficient resources for consensus messages to ensure progress.
For this, we adopt two optimizations.
First, we prioritize the transmission and processing of consensus messages.
Second, we use a token-based limiter to limit the sending rate of data messages: every data message (i.e., microblock) needs a token to be sent out, and tokens are refilled at a configurable rate.
This ensures that the network resources will not be overtaken by data messages.
\fangyu{The above optimizations are specially designed for Stratus and are only used in Stratus-based implementations.
We did \emph{not} use these optimizations in non-Stratus protocols in our evaluation since they may negatively effect those protocols.}

\begin{table}[t]
  \caption{Summary of evaluated protocols.}
  \label{table:protocols}
\begin{center}
\footnotesize
\begin{tabular}{ | m{2cm} | m{5.8cm} | } 
  \hline
  \textbf{Acronym}& \textbf{Protocol description} \\ 
  \hline
  \textbf{N-HS} & Native HotStuff
  without a shared mempool \\ 
  
  \fangyu{\textbf{N-PBFT}} & \fangyu{Native PBFT}
  without a shared mempool \\ 
  
  \textbf{SMP-HS} & HotStuff integrated with a simple shared mempool \\ 
  
  \textbf{SMP-HS-G} & SMP-HS with gossip instead of broadcast \\ 
  
  \fangyu{\textbf{SMP-HS-Even}} & \fangyu{SMP-HS with an even workload across replicas}\\ 
  \hline
  \textbf{S-HS} & HotStuff integrated with \textbf{Stratus (this paper)}\\ 

  \fangyu{\textbf{S-PBFT}} & \fangyu{PBFT integrated with \textbf{Stratus (this paper)}}\\ 
  \hline
  \textbf{Narwhal
  } & HotStuff based shared mempool \\ 
  \fangyu{\textbf{MirBFT}}
   & \fangyu{PBFT based multi-leader protocol} \\
  \hline
\end{tabular}
\end{center}
\end{table}

\section{Evaluation} \label{sec:evaluation}
Our evaluation answers the following questions.
\begin{itemize}[leftmargin=*]
    \item \textbf{Q1:} how does Stratus perform as compared to the alternative Shared Mempool implementations with a varying number of replicas? (Section~\ref{sec:eval-scalability})
    \item \textbf{Q2:} how do missing transactions caused by network asynchrony and Byzantine replicas affect the protocols' performance? (Section~\ref{sec:eval-byz})
    \item \textbf{Q3:} how does unbalanced load affect protocols' throughput? (Section~\ref{sec:eval-unbalance})
\end{itemize}

\subsection{Setup} \label{sec:eval-setup}

\noindent\textbf{Testbeds.}
We conducted our experiments on Alibaba Cloud \textsf{ecs.s6-c1m2.xlarge} instances\footnote{https://www.alibabacloud.com/help/en/doc-detail/25378.htm}.
Each instance has 4vGPUs and 8GB memory and runs Ubuntu server 20.04.
We ran each replica on a single ECS instance.
We performed protocol evaluations in LANs and WANs to simulate \textit{national} and \textit{regional} deployments, respectively~\cite{kauri}.
\fangyu{
LANs and WANs are typical deployments of permissioned blockchains and permissionless blockchains that run a BFT-based PoS consensus protocol~\cite{Tendermint,Dapper}.
}
In LAN deployments, a replica has up to 3 Gb/s of bandwidth and inter-replica RTT of less than 10 ms.
For WAN deployments, we use NetEm~\cite{netem} to simulate a WAN environment with 100 ms inter-replica RTT and 100 Mb/s replica bandwidth.

\vspace{1mm}
\noindent\textbf{Workload.}
Clients are run on 4 instances with the same specifications. \fangyu{Each client concurrently sends multiple transactions to different replicas}.
Bamboo's benchmark provides an in-memory key-value store backed by the protocol under evaluation.
Each transaction is issued as a simple key-value set operation submitted to a single replica.
Since our focus is on the performance of the consensus protocol with the mempool, we do not involve application-specific verification (including signatures) and execution (including disk IO operations) of transactions in our evaluation.
We measure both throughput and latency on the server side.
The latency is measured between the moment a transaction is first received by a replica and the moment the block containing it is committed.
We avoid end-to-end measurements to exclude the impact of the network delay between a replica and a client.
\fangyu{Each data point is obtained when the measurement is stabilized (sampled data do not vary by more than 1\%) and is an average over 3 runs.}
In our experiments, workloads are evenly distributed across replicas except for the last set of experiments (Section~\ref{sec:eval-unbalance}), in which we create skewed load to evaluate load balancing.

\vspace{1mm}
\noindent\textbf{Protocols.}
We evaluate the performance of a wide range of protocols (Table~\ref{table:protocols}).
\fangyu{We use native HotStuff and PBFT with the original mempool as the baseline, denoted as (N-HS and N-PBFT, respectively).
All of our implementations of HotStuff are based on the Chained-HotStuff (three-chain) version from the original paper~\cite{hotstuff}, in which pipelining is used and leaders are rotated for each proposal.
Our implementation of PBFT shares the same chained blockchain structure as Chained-HotStuff for a fair comparison.
}
We also compare against a version of HotStuff with a basic shared mempool with best-effort broadcast and fetching (denoted as SMP-HS).
\fangyu{Finally, we equip HotStuff and PBFT with our Stratus Mempool, denoted as S-HS and S-PBFT, respectively.}
We also implemented a gossip-based shared mempool (distributing microblocks via gossip), denoted by SMP-HS-G, to evaluate load balancing and compare it with S-HS.
All protocols are implemented using the same Bamboo code base for a fair comparison.
The sampling parameter $d$ is set to $1$ by default.
\fangyu{
This is because $d=1$ allows the busy sender to randomly pick exactly one replica without comparing workload status between others.
When we gradually increase $d$, the chance of selecting a less busy replica increases significantly.
However, increasing $d$ also incurs overhead.
In our experiments (Section~\ref{sec:eval-unbalance}) we show that $d=3$ exhibits the best performance.
}

We also compare against Narwhal\footnote{Available at https://github.com/facebookresearch/narwhal/}, which uses a shared mempool with reliable broadcast.
Narwhal is based on HotStuff and splits functionality between workers and primaries, responsible for transaction dissemination and consensus, respectively.
To fairly compare Narwhal with Stratus, we let each primary have one worker and locate both in one VM instance.
\fangyu{As another baseline, we compare our protocols with MirBFT~\cite{mir-bft}\footnote{Available at https://github.com/hyperledger-labs/mirbft/tree/research}, a state-of-the-art multi-leader protocol.
All replicas act as leaders in an epoch for fair comparison.}

\subsection{Scalability}\label{sec:eval-scalability}

\arxiv{
In the first set of experiments, we explore the impact of batch sizes on S-HS and then we evaluate the scalability of protocols.
These experiments are run in a common BFT setting in which less than one-third of replicas remain silent.
Since our focus is on normal-case performance, view changes are not triggered in these experiments unless clearly stated.
}

\begin{figure}[t]
    \centering
    \includegraphics[width=0.4\textwidth]{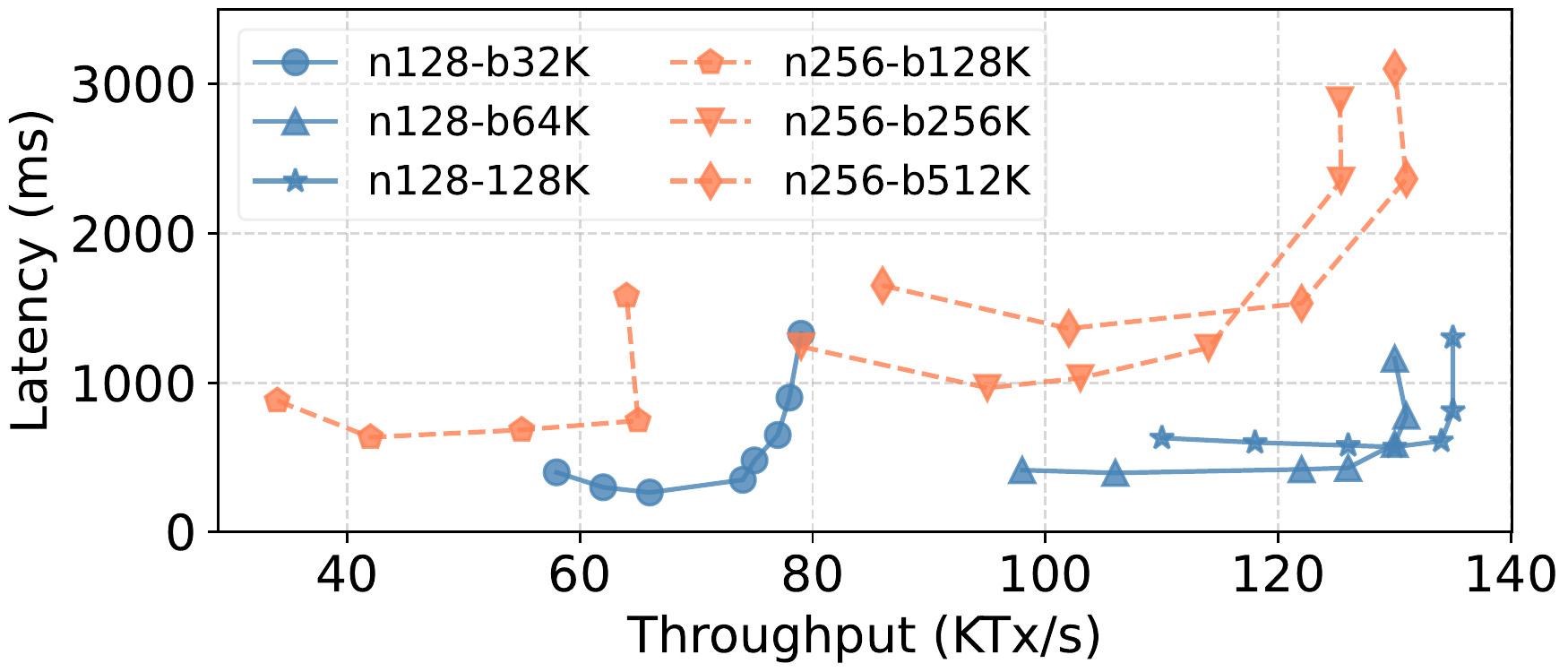}
    \caption{Throughput vs. latency with $128$ and $256$ replicas for S-HS. The batch size varies from 32KB to 512KB. The transaction payload is $128$ bytes.}
    \vspace{-0.5cm}
    \label{fig:batch-size}
\end{figure}

\vspace{1mm}
\arxiv{
\noindent\textbf{Picking a proper batch size.}
Batching more transactions in a microblock can increase throughput since the message cost is better amortized (e.g., fewer acks).
However, batching also leads to higher latency since it requires more time to fill a microblock.
In this experiment, we study the impact of batch size on Stratus (S-HS) and pick a proper batch size for different network sizes to balance throughput and latency.
}

\arxiv{
We deploy Stratus-based HotStuff (S-HS) in a LAN setting with $N=128$ and $N=256$ replicas, respectively.
For $N=128$, we vary the batch size from 32KB to 128KB, while for $N=256$, we vary the batch size from 128KB to 512KB.
We denote each pair of settings as the network size followed by the batch size.
For instance, the network size of $N=128$ with the batch size of 32KB bytes is denoted as $n128-b32$K.
We use the transaction payloads of $128$ bytes (commonly used in blockchain systems~\cite{bitcoin,diem}).
We gradually increase the workload until the system is saturated, i.e., the workload exceeds the maximum system throughput, resulting in sharply increasing delay.
}

\arxiv{
The results are depicted in Figure~\ref{fig:batch-size}.
We can see that as the batch size increases, the throughput improves accordingly for both network sizes.
However, the throughput gain of choosing a larger batch size is reduced when the batch size is beyond 64KB (for $N=128$) and 256KB (for $N=256$).
Also, we observe that a larger network requires a larger batch size for better throughput.
This is because large batch size amortizes the overhead of PAB (fewer acks).
But, a larger batch size leads to increased latency (as we explained previously).
We use the batch size of 128KB for small networks ($N\leq128$), the batch size of 256KB for large networks ($N\geq256$), and a $128$-byte transaction payload in the rest of our experiments.
As long as a replica accumulates sufficient transactions (reaching the batch size), it produces and disseminates a microblock.
If the batch size is not reached before a timeout (\SI{200}{\ms} by default), all the remaining transactions will be batched into a microblock.
We also find that proposal size (number of microblock ids included in a proposal) does not have obvious impact on the performance as long as a proper batch size (number of transactions included in a microblock) is chosen.
Therefore, we do not set any constraint on proposal size.
The above settings also apply in SMP-HS and SMP-HS-G.
}

\begin{figure}[t]
\centering
\begin{subfigure}{0.45\textwidth}
    \includegraphics[width=\textwidth]{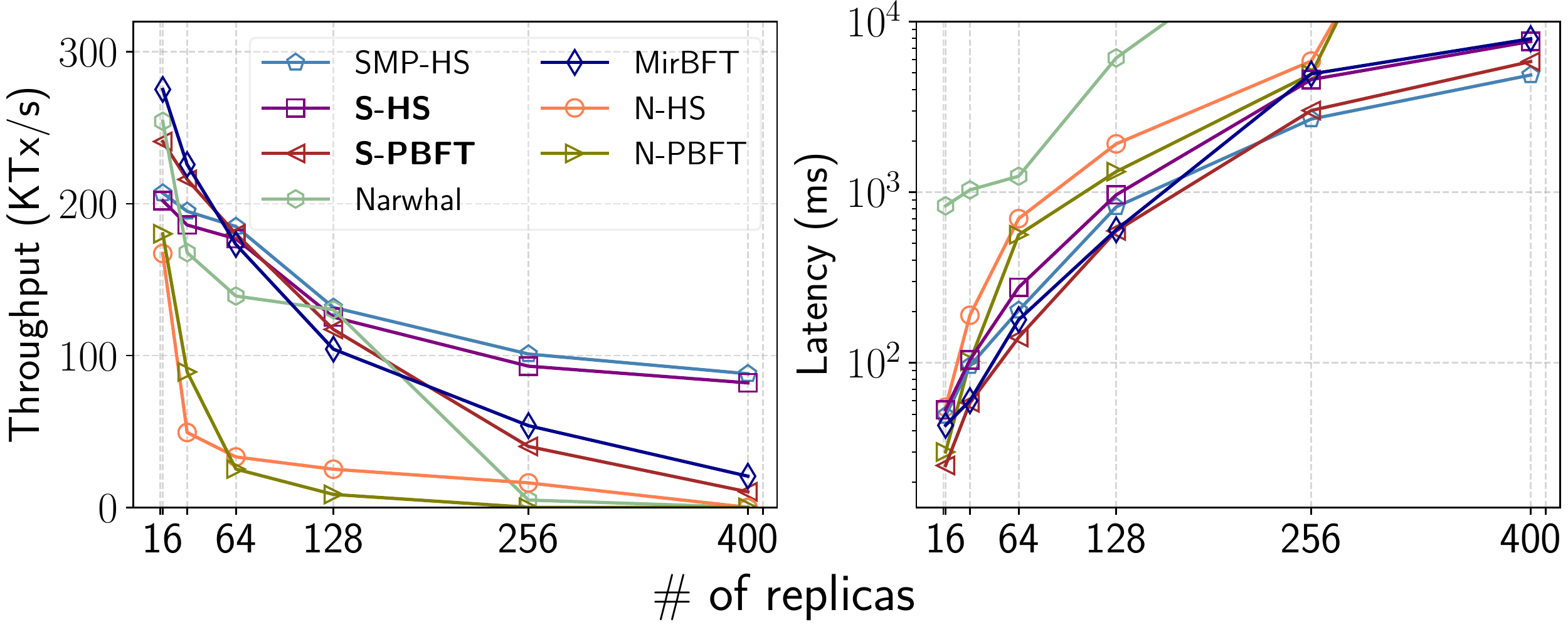}
    \caption{LAN evaluation.}
    \label{fig:scalability-lan}
\end{subfigure}
\hfill
\begin{subfigure}{0.45\textwidth}
    \includegraphics[width=\textwidth]{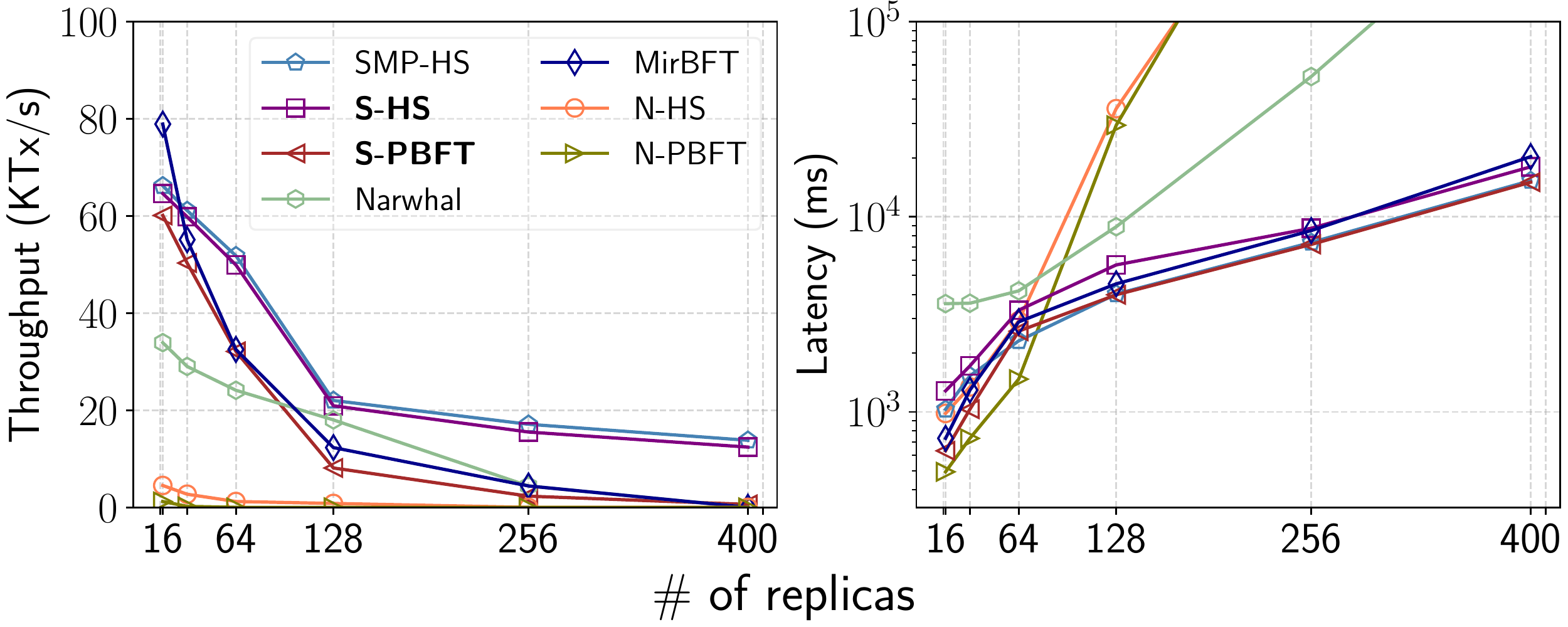}
    \caption{WAN evaluation.}
    \label{fig:scalability-wan}
\end{subfigure}
\caption{The throughput (left) and latency (right) of protocols in both LAN and WAN with increasing number of replicas. We use $128$-byte payload and 128KB batch size.}
\label{fig:scalability}
\end{figure}

We evaluate the scalability of the protocols by increasing the number of replicas from 16 to 400.
\fangyu{We use N-HS, N-PBFT, SMP-HS, S-PBFT, Narwhal, and MirBFT for comparison and run experiments in both LANs and WANs.
We gradually increase the workload until the system is saturated, i.e., the workload exceeds the maximum system throughput, resulting in sharply increasing delay.
}

\fangyu{We use a batch size of 256KB and $128$-byte transaction payload, which gives 2,000 transactions per batch, for Stratus-based protocols throughout our experiments.}
We find that proposal size (number of microblock ids included in a proposal) does not have an obvious impact on performance as long as we choose a proper batch size (number of transactions in a microblock).
Therefore, we do not constrain proposal size.
\fangyu{For every protocol we use a microblock/proposal size settings that maximizes the protocol's performance.
We omit experimental results that explore these settings due to space constraints.}


Figure~\ref{fig:scalability} depicts the throughput and latency of the protocols with an increasing number of replicas in LANs and WANs.
We can see that protocols using the shared mempool (SMP-HS, S-HS, S-PBFT, and Narwhal) or relying on multiple leaders outperform the native HotStuff and Streamlet (N-HS and N-PBFT) in throughput in all experiments.
Previous works~\cite{hotstuff,leopard,narwal} have also shown that the throughput/latency of N-HS decreases/increases sharply as the number of replicas increases, and meaningful results can no longer be observed beyond 256 nodes.
Although Narwhal outperforms N-HS due to the use of a shared mempool, it does not scale well since it employs the heavy reliable broadcast primitive.
As shown in~\cite{narwal}, Narwhal achieves better scalability only when each primary has multiple workers that are located in different machines.
\fangyu{MirBFT has higher throughput than S-HS when there are fewer than 16 replicas.
This is because Stratus imposes a higher message overhead than PBFT. However, MirBFT's performance drops faster than S-HS because of higher message complexity.
MirBFT is comparable to S-PBFT because they have the same message complexity.
The gap between them is due to implementation differences.
}

SMP-HS and S-HS show a slightly higher latency than N-HS when the network size is small ($<16$ in LANs and $<32$ in WANs).
This is due to batching.
They outperform the other two protocols in both throughput and latency when the network size is beyond 64 and show flatter lines in throughput as the network size increases.
The throughput of SMP-HS and S-HS achieve $5\times$ throughput when $N=128$ as compared to N-HS, and this gap grows with network size.
Finally, SMP-HS and S-HS have similar performance, which indicates that the use of PAB incurs negligible overhead, which is amortized by a large batch size.

\begin{table}[ht!]
  \caption{Outbound bandwidth consumption comparison with $N=64$ replicas.
  The bandwidth of each replica is throttled to 100 Mb/s.
  The results are collected when the network is saturated.}
  \label{table:consumption}
\begin{center}
\footnotesize
\begin{tabular}{ | c | c | c | c | c | } 
  \hline
  \multicolumn{2}{|c|}{Role/Messages} & N-HS & SMP-HS & \textbf{S-HS (this paper)}\\
  \hline \hline  
  \multirow{3}{*}{Leader} & Proposals & 75.4 & 4.7 & 9.8\\
  \cline{2-5}
  & Microblocks & N/A & 50.5 & 50.3\\
  \cline{2-5}
  & \textbf{SUM} & \textbf{75.4} & \textbf{55.2} & \textbf{60.1}\\
  \hline\hline
  \multirow{4}{*}{Non-leader} & Microblocks & N/A & 50.4 & 50.3\\
  \cline{2-5}
  & Votes & 0.5 & 2.5 & 2.4\\
  \cline{2-5}
  & Acks & N/A & N/A & 4.7\\
  \cline{2-5}
  & \textbf{SUM} & \textbf{0.5} & \textbf{52.9} & \textbf{57.4}\\
  \hline
\end{tabular}
\end{center}
\end{table}

\noindent\textbf{Bandwidth consumption.} We evaluate the outbound bandwidth usage at the leader and the non-leader replica in N-HS, SMP-HS, and S-HS.
We present the results in Table~\ref{table:consumption}.
We can see that the communication bottleneck in N-HS is at the leader, while the bandwidth of non-leader replicas is underutilized.
In SMP-HS and S-HS, the bandwidth consumption between leader replicas and non-leader replicas are more even, and the leader bottleneck is therefore alleviated.
We observe that S-HS adds around 10\% overhead on top of SMP-HS due to the use of PAB.
Next, we show that this overhead is worthwhile as it provides availability insurance.
\edits{We also observe that around 40\% of bandwidth remains unused.
This is because chain-based protocols are bounded by latency: each proposal goes through two rounds of communication (one-to-all-to-one).
We consider out-of-order processing of proposals for better network utilization as important future work.}

\subsection{Impact of Missing Transactions}~\label{sec:eval-byz}

\noindent Recall that in \textbf{Problem-I} (Section~\ref{sec:challenges}), a basic shared mempool with best-effort broadcast is subject to \textit{missing transactions}.
In the next set of experiments, we evaluate the throughput of SMP-HS and S-HS under a period of network asynchrony and Byzantine attacks.

\vspace{1mm}
\noindent \textbf{Network asynchrony.}
\fangyu{During network asynchrony, a proposal is likely to arrive before some of referenced transactions (i.e., missing transactions), which negatively impacts performance.
The point of this experiment is to show that Stratus-based protocols can make progress during view-changes and are more resilient to network asynchrony.}

We ran an experiment in a WAN setting, during which we induce a period of network fluctuation via NetEm.
The fluctuation lasts for \SI{10}{\s}, during which network delays between replicas fluctuate between \SI{100}{\ms} and \SI{300}{\ms} for each message (i.e., \SI{200}{\ms} base with \SI{100}{\ms} uniform jitter).
We set the view-change timer to be \SI{1000}{\ms}.
We keep the transaction rate at 25KTx/s without saturating the network.

We ran the experiment 10 times and each run lasts 30 seconds.
We show the results in Figure~\ref{fig:async}.
During the fluctuation, the throughput of SMP-HS drops to zero.
This is because missing transactions are fetched from the leader, which causes congestion at the leader.
As a result, view-changes are triggered, during which no progress is made.
When the network fluctuation is over, SMP-HS slowly recovers by processing the accumulated proposals.
On the other hand, S-HS makes progress at the speed of the network and no view-changes are triggered.
This is due to the \textbf{PAB-Provable Availability} property: no missing transactions need to be fetched on the critical consensus path.

\begin{figure}[t]
    \centering
    \includegraphics[width=0.45\textwidth]{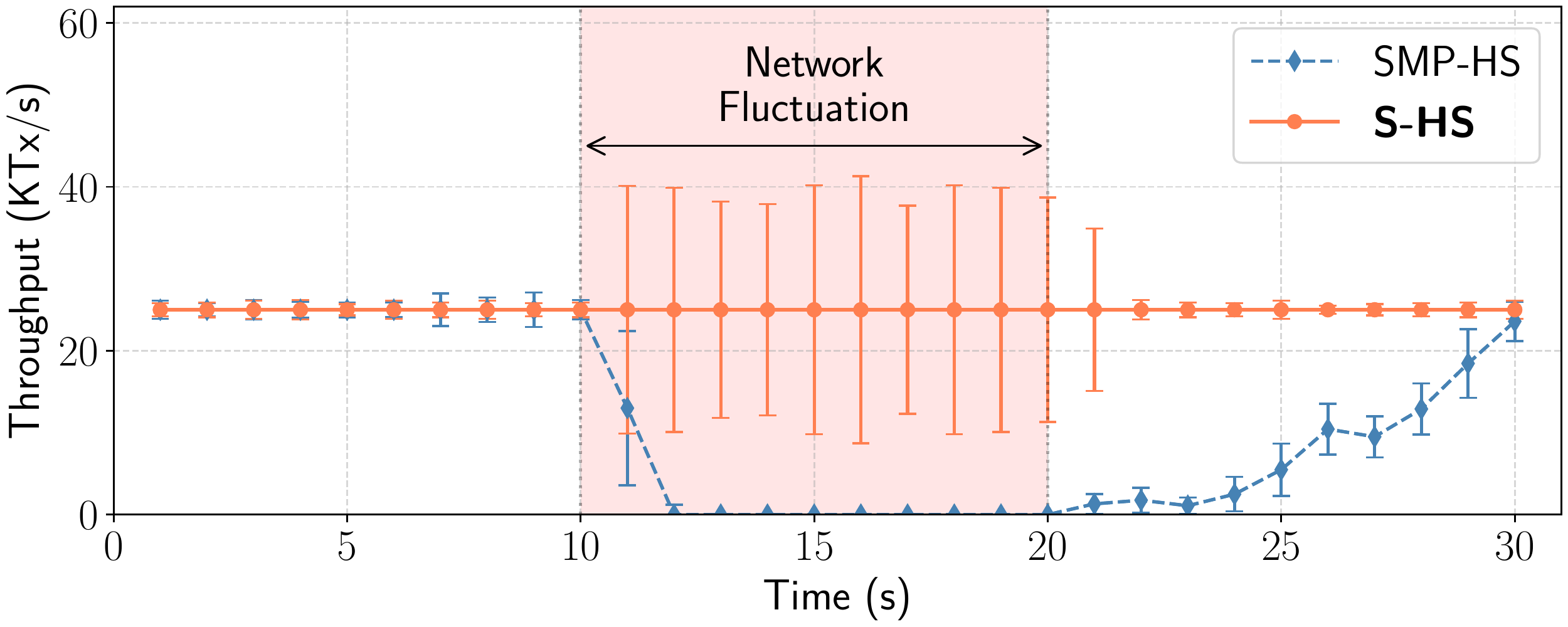}
    \caption{Delay is injected at time \SI{10}{\s} and lasts for \SI{10}{\s}.
    The transaction rate is 25KTx/s. Each point is averaged over 10 runs.}
    \label{fig:async}
\end{figure}

\vspace{1mm}
\noindent \textbf{Byzantine senders.}
The attacker's goal in this scenario is to overwhelm the leader with many missing microblocks.

The strategies for each protocol are described as follows.
In SMP-HS, Byzantine replicas only send microblocks to the leader (Figure~\ref{fig:mempool-set}).
In S-HS, Byzantine replicas have to send microblocks to the leader and to at least $f$ replicas to get proofs.
Otherwise, their microblocks will not be included in a proposal (consider the leader is correct).
In this experiment, we consider two different quorum parameters for PAB (see Section~\ref{sec:discussion}), $f+1$ and $2f+1$ (denoted by S-HS-f and S-HS-2f, respectively).
These variants will explain the tradeoff between throughput and latency.
We ran this experiment in a LAN setting with $N=100$ and $N=200$ replicas (including the leader).
The number of Byzantine replicas ranged from $0$ to $30$ ($N=100$) and $0$ to $60$ ($N=200$).

\begin{figure}[t]
\centering
\begin{subfigure}{0.45\textwidth}
    \includegraphics[width=\textwidth]{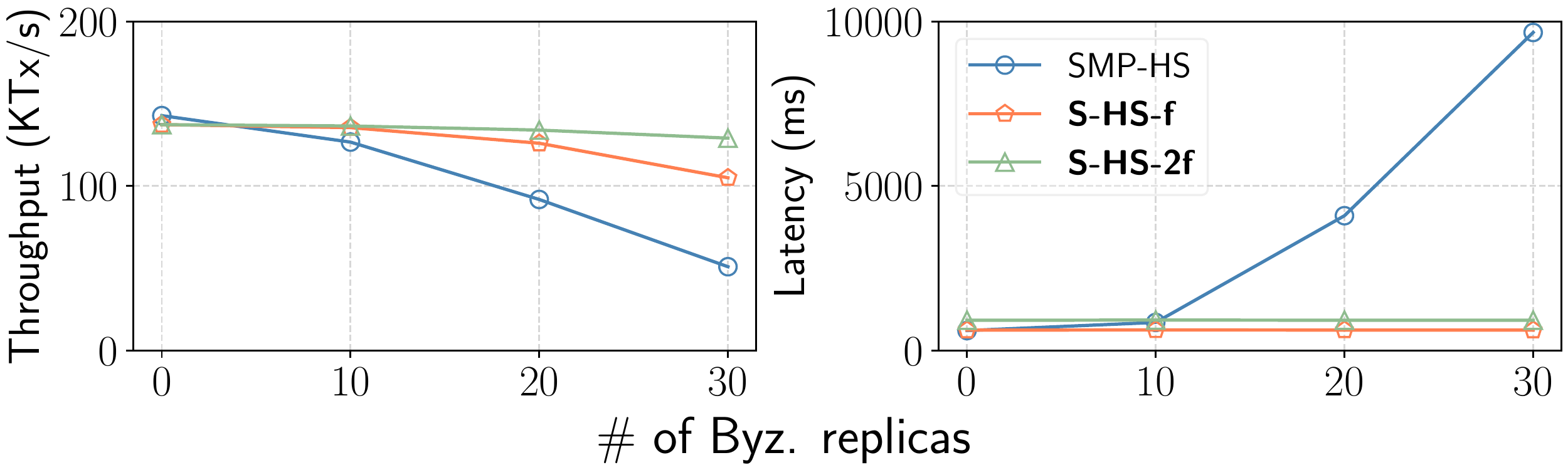}
    \caption{100 total replicas with 0 to 30 Byz. ones.}
    \label{fig:byz-100}
\end{subfigure}
\hfill
\begin{subfigure}{0.45\textwidth}
    \includegraphics[width=\textwidth]{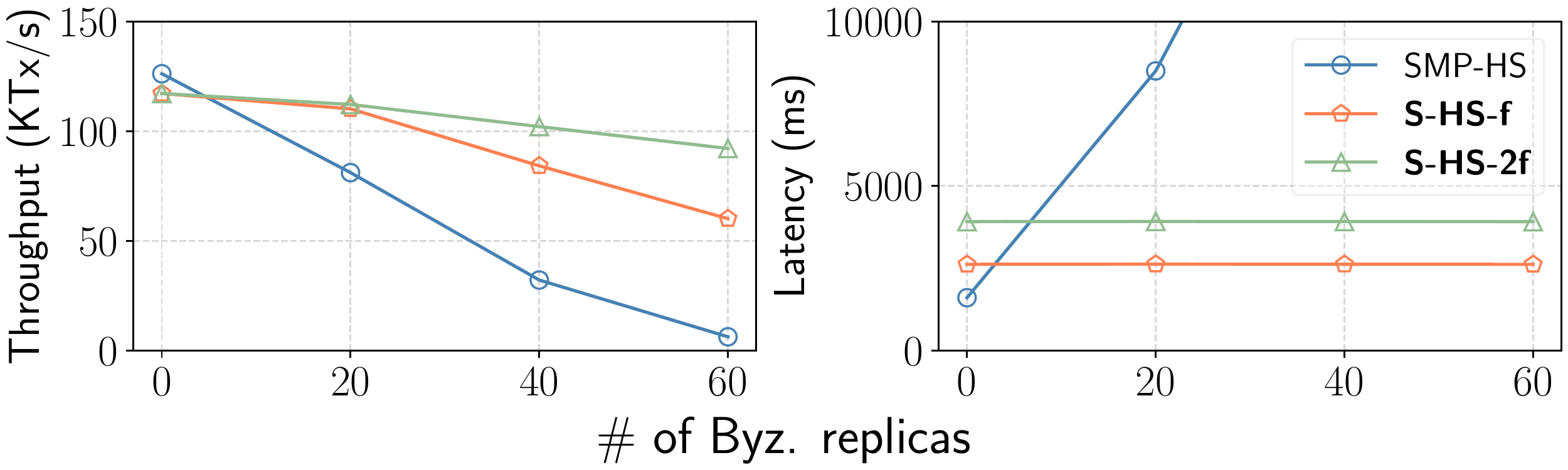}
    \caption{200 total replicas with 0 to 60 Byz. ones.}
    \label{fig:byz-200}
\end{subfigure}
\caption{Performance of SMP-HS and S-HS with different quorum parameters (S-HS-d1 and S-HS-d2) and increasing Byzantine replicas.}
\label{fig:byz}
\end{figure}

Figure~\ref{fig:byz} plots the results.
As the number of Byzantine replicas increases, the throughput/latency of SMP-HS decreases/increases sharply.
This is because replicas have to fetch missing microblocks from the leader before processing a proposal.
We also observe a slight drop in throughput of S-HS.
The reason is that only background bandwidth is used to deal with missing microblocks.
The latency of S-HS remains flat since the consensus will never be blocked by missing microblocks as long as the leader provides correct proofs.
In addition, we notice that Byzantine behavior has more impact on larger deployments.
With $N=200$ replicas, the performance of SMP-HS decreases significantly.
The throughput is almost zero when the number of Byzantine replicas is $60$ and the latency surges when there are more than $20$ Byzantine replicas.
Finally, S-HS-2f has better throughput than S-HS-f at the cost of higher latency as the number of Byzantine replicas increases.
The reason is that with a larger quorum size, fewer microblocks need to be fetched.
However, a replica needs to wait for more acks to generate available proofs.

\subsection{Impact of Unbalanced Workload}\label{sec:eval-unbalance}
Previous work~\cite{ethna,info-propagation,txprobe,Miller2015DiscoveringB} has observed that node degrees in large-scale blockchains have a power-law distribution.
As a result, most clients send transactions to a few popular nodes, leading to unbalanced workload (\textbf{Problem-II} in Section~\ref{sec:challenges}).
In this experiment, we vary the ratio of workload to bandwidth by using identical bandwidth for each replica but skewed workloads across replicas.
We use two Zipfian parameters~\cite{zipf-gen}, Zipf1 ($s=1.01, v=1$) and Zipf10 ($s=1.01, v=10$), to simulate a highly skewed workload and a lightly skewed workload, respectively.
We show the workload distributions in Figure~\ref{fig:distribution}.
\fangyu{For example, when $s=1.01$ and there are 100 replicas, $10\%$ of the replicas will receive over $85\%$ of the load.}

\begin{figure}[h]
    \centering
    \includegraphics[width=0.45\textwidth]{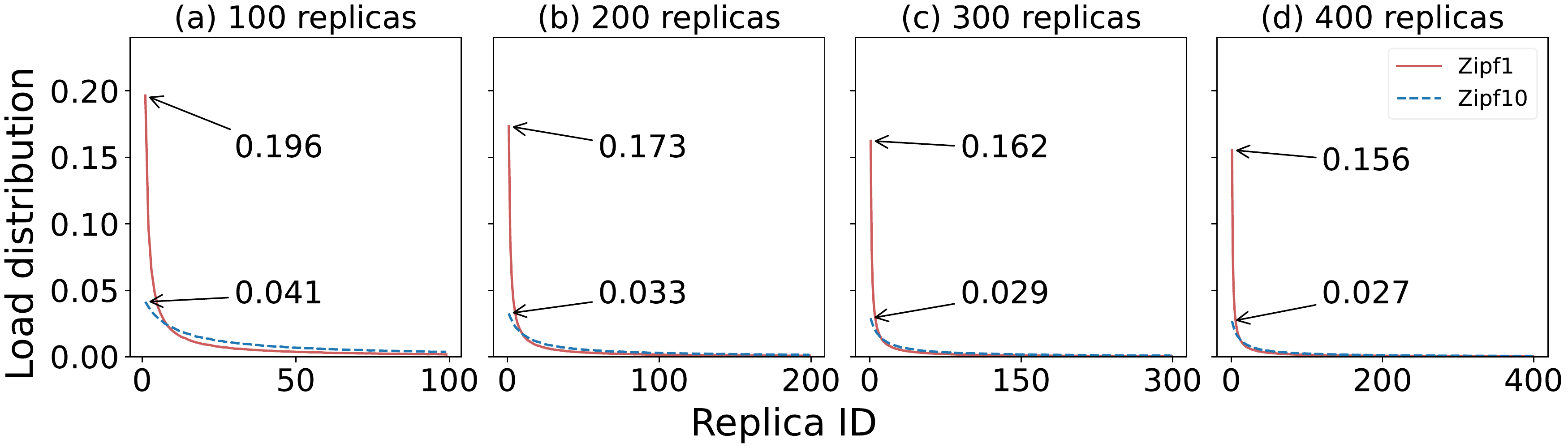}
    \caption{Workload distribution with different network sizes and Zipfian parameters.}
    \label{fig:distribution}
\end{figure}

We evaluate load-balancing in Stratus using the above distributions in a WAN setting.
Stratus samples $d$ replicas to select the least loaded one as the proxy, we consider $d=1,2,3$, denoted by S-HS-d1, S-HS-d2, and S-HS-d3, respectively.
We also use SMP-HS-G, HotStuff with a gossip-based shared mempool for comparison.
We set the gossip fan-out parameter to $3$.

\begin{figure}[t]
    \centering
    \includegraphics[width=0.45\textwidth]{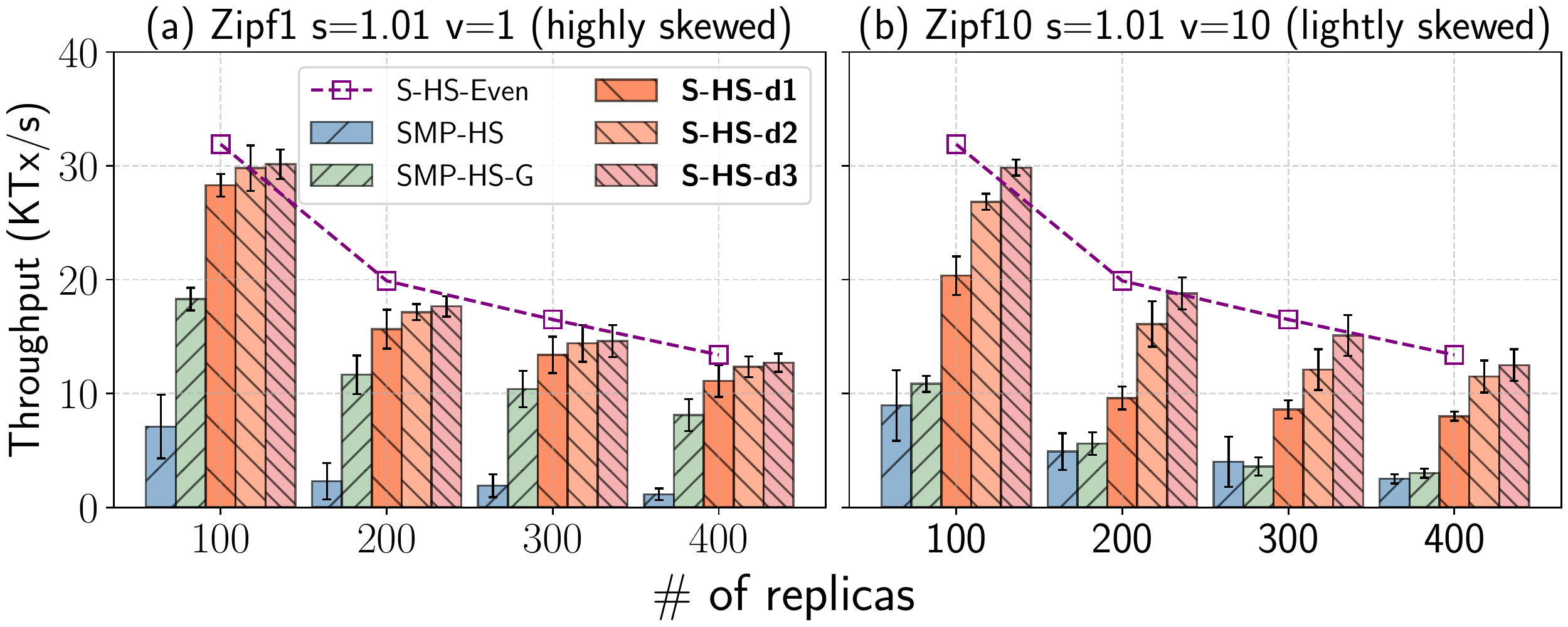}
    \caption{Throughput with different workload distribution.}
    \label{fig:throughput}
\end{figure}

Figure~\ref{fig:throughput} shows protocols' throughput.
We can see that S-HS-dX outperforms SMP-HS and SMP-HS-G in all experiments.
S-HS-dX achieves $5\times$ ($N=100$) to $10\times$ ($N=400$) throughput with Zipf1 as compared with SMP-HS.
SMP-HS-G does not scale well under a lightly skewed workload (Zipf10) due to the message redundancy.
We also observe that S-HS-dX achieves the best performance when $d=3$, while the gap between different $d$ values is not significant. 

\section{Discussion}\label{sec:discussion}



\fangyu{
\textbf{Attacks on PAB.}
Byzantine replicas can create availability proofs and send them to fewer than $f$ replicas.
If the leader is correct, then a valid proposal is proposed with microblock ids and their availability proofs.
Using these, replicas can recover if a referenced microblock is missing.
The microblocks with missing proofs will be discarded after a timeout. 

Now consider a Byzantine leader that includes microblocks without availability proofs into a proposal.
This will trigger a view-change, which will replace the leader.
In some PoS blockchains~\cite{Tendermint,diem}, such leaders are also slashed.}

\fangyu{
\textbf{Attacks on load balancing.}
A Byzantine sender can try to to congest the network by sending identical microblocks to multiple proxies.
To mitigate this attack, we propose a simple solution.
When a busy replica $r$ decides on $r^{\prime}$ as the proxy, it forwards the microblock $mb$ to $r^{\prime}$ along with a message $\rho$ that contains $r$'s signature over $mb.id$ concatenated with $r^{\prime}$'s identity.
Then, $r^{\prime}$ broadcast $mb$ along with $\rho$ using PAB.
This allows other replicas to check if a microblock by the same sender is broadcast by different proxies.
Once detected, a replica can reject microblocks from this sender or report this behavior by sending evidence to the other replicas.
If the proxy fails to complete PAB, the original sender either broadcasts the microblock by itself or waits for a timeout to garbage collect the microblock.
}

\fangyu{
A malicious replica can pretend to be busy and forward its load to other replicas. This can be addressed with an incentive mechanism: a replica that produced the availability proof for a microblock using PAB is rewarded.
This information is verifiable because the availability proofs for each microblock are in the proposal and will be recorded on the blockchain if the proposal is committed.
In addition, to prevent a malicious senders from overloading a proxy, the proxy can set a limit on its buffer, and reject extra load.
}

\arxiv{
\textbf{Re-configuration.} Stratus can be extended to support adding or removing replicas.
For example, Stratus can subscribe to re-configuration events from the consensus engine.
When new replicas join or leave, Stratus will update its configuration.
Newly joined replicas may then fetch stable microblocks (i.e., ids with available proofs) to catch up.
}

\arxiv{
\textbf{Garbage collection.} To ensure that transactions remain available, replicas may have to keep the microblocks and relevant meta-data (e.g., acks) in case other replicas fetch them.
To garbage-collect these messages, the consensus protocol should inform Stratus that a proposal is committed and the contained microblocks can then be garbage collected.
}

\section{Conclusion and Future Work} \label{sec:conclusion}

We presented a shared mempool abstraction that resolves the leader bottleneck of leader-based BFT protocols.
We designed Stratus, a novel shared mempool protocol to address two challenges: missing transactions and unbalanced workloads.
Stratus overcomes these with an efficient provably available broadcast (PAB) and a load balancing protocol. For example, Stratus-HotStuff throughput is 5$\times$ to 20$\times$ higher than native HotStuff. 
In our future work, we plan to extend Stratus to multi-leader BFT protocols.

\balance
\bibliographystyle{unsrt}
\bibliography{stratus}

\clearpage
\appendix

\section{Analysis}
In this section, we theoretically reveal the leader bottleneck of leader-based BFT protocols (LBFT) and then show how shared mempool addresses the issue. 
We consider the ideal performance, i.e., all replicas are honest and the network is synchronous.
We assume that the ideal performance is limited by the available processing capacity of each replica, denoted by $C$. 
For simplicity, we further assume that transactions have the same size $B$ (in bits).
We use $T_{max}$ to denote the maximum throughput, i.e., number of transactions per second.
We use $W_{l}$ (resp. $W_{nl}$) to denote the workload of the leader (resp. a non-leader replica) for confirming a transaction. 
Furthermore, we have
\[T_{max} = \min\left \{\frac{C}{W_{l}}, \frac{C}{W_{nl}}\right \}.\]
Since each replica has to receive and process the transaction once, we have $W_{l}, W_{nl} \geq B$. 
Besides, due to the protocol overhead, we have $W_{l}, W_{nl} > B$. As a result, $T_{max} < C/B$.
In other words, $C/B$ is the upper bound of the maximum throughput of any BFT protocol. 

\subsection{Bottleneck of LBFT Protocols}\label{sec:leader-bottleneck-analysis}
In LBFT protocols, when making a consensus of a transaction, the leader is in charge of disseminating it to other $n-1$ replicas, while each non-leader replica proceeds it from the leader. 
Hence, the workloads of proceeding with the transaction for the leader and a non-leader replica are $W_l = B (n-1)$ and $W_{nl} = B$, respectively.
Furthermore, we have 
\[T_{max} = \min\left \{\frac{C}{B (n-1)}, \frac{C}{B}\right \} = \frac{C}{B (n-1)}.\]
The equation shows that with the increase of replicas, the maximum throughput of LBFT protocols will drop proportionally. 
Note that protocol overhead is not considered, which makes it easier to illustrate the unbalanced loads between the leader and non-leader replicas and to show the leader bottleneck. 

Next, we take PBFT~\cite{pbft} as a concrete example to show more details of the leader bottleneck. 
In PBFT the agreement of a transaction involves three phases: the pre-prepare, prepare, and commit phases. In particular, the leader first receives a transaction from a client and then disseminates the transaction to all other $n-1$ replicas in the pre-prepare phase.
In prepare and commit phases, each replica broadcasts their vote messages and receives all others' vote messages for reaching consensus.\footnote{In the implementation, the leader does not need to broadcast its votes in the prepare phase since the proposed transaction could represent the vote message.}
Let $\sigma$ denote the size of voting messages. 
The workloads for the leader and a non-leader replica are $ W_l = n B + 4 (n-1) \sigma$ and $ W_{nl} = B + 4 (n-1) \sigma$, respectively.
Finally, we can derive the maximum throughput of PBFT as
\[T_{max} = \min\left \{\frac{C}{nB + 4 (n-1) \sigma}, \frac{C}{B + 4 (n-1) \sigma}\right \}.\]
The equations show that both the dissemination of the transaction and vote messages limit the throughput. 
Besides, we can see that when processing a transaction, each replica has to process  $4 (n-1)$ vote messages, which leads to high protocol overhead. 
To address this, multiple transactions can be batch into a proposal (e.g., forming a block) to amortize the protocol overhead.  
For example, let $K$ denote the size of a proposal, and the maximum throughput of PBFT when adopting batch strategy is
\[T_{max} = \frac{K}{B} \times \min\left \{\frac{C}{nK + 4 (n-1)\sigma}, \frac{C}{K + 4 (n-1) \sigma}\right \}.\]
When $K$ is large (i.e., $K \gg \sigma$), we have $\frac{C}{n K + 4 (n-1) \sigma} \approx \frac{C}{n K}$ and $T_{max} = \frac{C}{nB}$. This shows that the maximum throughput drops with the increasing number of replicas, and the dissemination of the proposal by the leader is still the bottleneck. 
In other words, batching strategy cannot address the scalability issues of LBFT protocols. 
What is more, several state-of-the-art LBFT protocols such as HotStuff~\cite{hotstuff} achieve the linear message complexity by removing the $(n-1)$ factor from the $(n-1)\sigma$ overhead of non-leader replicas. 
However, this also cannot address the scalability issue since the proposal dissemination for the leader is still the dominating component. 

\subsection{Analysis of Using Shared Mempool}\label{sec:shared-mempool-analysis}
To address the leader bottleneck of LBFT protocols, our solution is to decouple the transaction dissemination with a consensus algorithm, by which dissemination workloads can be balanced among all replicas, leading to better utilization of replicas' processing capacities. 
In particular, to improve the efficiency of dissemination, transactions can be batched into microblocks, and replicas disseminate microblocks to each other. 
Each microblock is accompanied by a unique identifier, which can be generated by the hash function. 
Later, after a microblock is synchronized among replicas, the leader only needs to propose an identifier of the microblock. 
Since the unique mapping between identifiers and microblocks, ordered identifiers lead to a sequence of microblocks, which further determines a sequence of transactions. 

Next, we show how the above decoupling idea can address the leader bottleneck.  
We use $\gamma$ to denote the size of an identifier and $\eta$ to denote the size of a microblock.  
Given a proposal with the same size $K$, it can include $K/\gamma$ identifiers. Each identifier represents a microblock with $\eta/B$ transactions. Hence, a proposal represents $\frac{K}{\gamma} \times \frac{\eta}{B}$ transactions. 
As said previously, the $K/\gamma$ microblocks are disseminated by all non-leader replicas, so each non-leader replica has to disseminate $K/(\gamma (n-1))$ microblocks to all other replicas.
Correspondingly, each replica (including the leader) can receive $K/(\gamma (n-1))$ microblocks from $n-1$ non-leader replicas. 
Hence, the workload for the leader is 
\[W_l = (n-1) \frac{K \eta}{\gamma (n-1)} + (n-1)K = \frac{K\eta}{\gamma} + (n-1)K,\]
where $(n-1)K$ is the workload for disseminating the proposal. Similarly, the workload for a non-leader replica is 
\[W_l = n \frac{K\eta}{\gamma (n-1)} + (n-2) \frac{K\eta}{\gamma (n-1)} + K = \frac{2K\eta}{\gamma} + K,\]
where $K$ is the workload for receiving a proposal from the leader.
Finally, we can derive the maximum throughput as  
\[T_{max} = \frac{K \eta}{\gamma B} \times \min\left \{\frac{C}{(K\eta)/\gamma + (n-1)K}, \frac{C}{(2K\eta)/\gamma + K}\right \}.\]
To make the throughout maximum, we can adjust $\eta$ and $\gamma$ to balance the workloads of the leader and non-leader replicas. This is $\frac{2K\eta}{\gamma} + K = \frac{K\eta}{\gamma} + (n-1)K$, and we have $\eta = (n-2) \gamma$.
Finally, we can obtain the maximum throughput is $T_{max} = \frac{ C(n-2)}{ B (2n-3)}$. 
Particularly, when $n$ is large, we have $T_{max} \approx \frac{C}{2B}$. The result is optimal since given a transaction, it has to be sent and received $n$ times (one for each replica), which leads to about $2 n B$ workload, and the total processing capacities of all replicas is $nC$.



\end{document}